%% file: neurips_2025.tex
\documentclass{article}

    \usepackage[final]{neurips_2025}

\usepackage[utf8]{inputenc} 
\usepackage[T1]{fontenc}    
\usepackage{hyperref}       
\usepackage{url}            
\usepackage{booktabs}       
\usepackage{amsfonts}       
\usepackage{nicefrac}       
\usepackage{microtype}      
\usepackage{xcolor}         
\usepackage{amsmath} 
\usepackage{amsthm}
\usepackage{natbib}
\usepackage{graphicx}
\usepackage{tabularx}  
\usepackage{makecell}
\usepackage{array,booktabs}
\usepackage{multirow} 
\usepackage{pifont}

\usepackage{subcaption}
\newtheorem{theorem}{Theorem}
\newtheorem{lemma}{Lemma}
\newtheorem{proposition}{Proposition}
\newtheorem{definition}{Definition}

\newcommand{\cmark}{\ding{51}}  
\newcommand{\xmark}{\ding{55}}  

\title{Beyond Last-Click: An Optimal Mechanism for Ad Attribution}

\author{%
  Nan An \\
  Gaoling School of Artificial Intelligence \\
  Renmin University of China \\
  Beijing, China \\
  \texttt{annan0425@ruc.edu.cn} \\
  \And
  Weian Li \\
  School of Software \\
  Shandong University \\
  Jinan, China \\
  \texttt{weian.li@sdu.edu.cn} \\
  \And
  Qi Qi\thanks{Corresponding author.} \\
  Gaoling School of Artificial Intelligence \\
  Renmin University of China \\
  Beijing, China \\
  \texttt{qi.qi@ruc.edu.cn} \\
  \And
  Changyuan Yu \\
  Baidu Inc. \\
  Beijing, China \\
  \texttt{yuchangyuan@baidu.com} \\
  \And
  Liang Zhang \\
  Gaoling School of Artificial Intelligence \\
  Renmin University of China \\
  Beijing, China \\
  \texttt{zhang.liang@ruc.edu.cn} \\
}

\begin{document}

\maketitle

\begin{abstract}
Accurate attribution for multiple platforms is critical for evaluating performance-based advertising. However, existing attribution methods rely heavily on the heuristic methods, e.g., Last-Click Mechanism (LCM) which always allocates the attribution to the platform with the latest report, lacking theoretical guarantees for attribution accuracy. In this work, we propose a novel theoretical model for the advertising attribution problem, in which we aim to design the optimal dominant strategy incentive compatible (DSIC) mechanisms and evaluate their performance.
We first show that LCM is not DSIC and performs poorly in terms of accuracy and fairness. To address this limitation, we introduce the Peer-Validated Mechanism (PVM), a DSIC mechanism in which a platform's attribution depends solely on the reports of other platforms. We then examine the accuracy of PVM across both homogeneous and heterogeneous settings, and provide provable accuracy bounds for each case. Notably, we show that PVM is the optimal DSIC mechanism in the homogeneous setting. Finally, numerical experiments are conducted to show that PVM consistently outperforms LCM in terms of attribution accuracy and fairness.

\end{abstract}

\input{Styles/Sources/intro}

\input{Styles/Sources/model}

\input{Styles/Sources/last-click}
\input{Styles/Sources/our_mechanism}
\input{Styles/Sources/experiment}

\input{Styles/Sources/conclusion}

\vspace*{-5pt}
\section*{Acknowledgment}
This work was supported by National Natural Science Foundation of China (No.62472428),
Public Computing Cloud, Renmin University of China, the fund for building world-class universities (disciplines) of Renmin University of China. 

\clearpage
\bibliographystyle{plain}
\bibliography{neurips_2025}

\clearpage
\appendix
\input{Styles/Sources/app}

\end{document}

%% file: Styles/Sources/intro.tex
\section{Introduction}
\label{sec:intro}
Online advertising has become the dominant force in the global advertising landscape, with expenditures projected to exceed \$790 billion in 2024---accounting for over 72\% of total ad spend---and continuing to grow at a consistent rate of more than 10\% annually. This substantial and growing capital investment calls for the development and application of robust methodologies to optimize budget allocation across diverse digital platforms. 

Advertising attribution, the process of assigning credit for user conversions (such as app downloads or product purchases) to the platforms that contributed to them, plays a central role in guiding these allocation decisions and has consequently garnered significant attention. In practice, attribution reflects a variety of design principles and business objectives. Methods range from simple heuristics such as first-click and time-decay attribution to data-driven approaches based on machine learning and causal inference. Among these, last-click attribution has become the industry default due to its simplicity and its practical relevance for measuring revenue-driven, bottom-of-funnel conversions.

Under last-click attribution, the platform that most recently interacted with the user receives full conversion credit. 
Because these credits directly determine performance metrics and future budget allocation, platforms have a strong incentive to manipulate the timing of their reports to appear last in the user’s interaction sequence. 
Such strategic behavior can distort attribution outcomes, overstating the influence of certain platforms even on its own terms.

This manipulation has become increasingly feasible in modern advertising ecosystems, especially when the advertiser does not control the landing page—such as app installations through app stores or purchases on major e-commerce platforms—where click events cannot be directly measured. 
In the past, advertisers relied on redirect-based tracking flows, where an intermediary measurement partner (MMP) logged the click before redirecting the user to the final landing page, thus providing an independent, verifiable timestamp. 
However, the industry has since shifted toward \emph{redirect-less} tracking paradigm, in which user navigation and click reporting are decoupled to improve latency and privacy. 
Without an intermediary verifier, advertisers now depend entirely on platform’s self-reported timestamps, making strategically timed reports both feasible and practically undetectable.\footnote{
This shift has been driven by privacy regulations such as the European Union’s \emph{General Data Protection Regulation} and Apple’s \emph{App Tracking Transparency} framework~\cite{EUGDPR2016,AppleATT2021,Kraft2024}, and the adoption of redirect-less systems including Google’s and Microsoft’s \emph{parallel tracking} and Apple’s \emph{SKAdNetwork}~\cite{GoogleAds2021,MicrosoftAds2021,AppleSKAd2021}.
}

Nevertheless, such strategic behavior has received limited attention in the academic literature. Most prior work on advertising attribution instead focuses on modeling platform contributions to conversions using increasingly sophisticated statistical or machine learning methods, under the assumption that platforms passively and truthfully report user interaction data.
In this paper, we initiate the study of advertising attribution from a mechanism design perspective, treating platforms as strategic agents that may misreport in order to maximize their assigned credit. Rather than proposing a new attribution philosophy, we work within the prevailing logic of last-click attribution and ask: \textit{How can we design an attribution mechanism such that platforms have no incentive to misreport, while still assigning credit to the platform with the true last click?}

\paragraph{Main Contribution}
To address the above question, we first model the advertising attribution scenario as a game-theoretic model in which multiple platforms strategically submit user interaction logs to compete for conversion credit. The advertiser then allocates credit according to a predefined attribution rule. In this model, our analysis focuses on characterizing dominant strategy incentive compatible (DSIC) mechanisms, and evaluating the performance of different attribution mechanisms, using two key metrics: accuracy and fairness. Accuracy measures the alignment between the assigned and true contributors, while fairness assesses whether each platform receives its deserved share of credit in expectation. Detailed results are presented in Table \ref{tab:pvm_settings}, with proofs in the full version.

\begin{table}[htbp]
\centering
\caption{Mechanism performance under different settings}
\resizebox{\textwidth}{!}{%
  \begin{tabular}{lcccc}
  \toprule
   & \textbf{DSIC} & \makecell{\textbf{Fair}}
   & \textbf{Accuracy (Homogeneous)} & \textbf{Accuracy (Heterogeneous)} \\
  \midrule
  \multirow{4}{*}{\textbf{LCM}}
    & \multirow{2}{*}{\xmark}
    & \multirow{2}{*}{\xmark}
    & \makecell[l]{$n=2$:\; $(2-\sqrt{2})^2\approx0.3431$}
    & \multirow{2}{*}{0} \\
    & 
    &
    & (Theorem~\ref{thm:last-click-two-homogeneous})
    &\\
    & \multirow{2}{*}{(Proposition~\ref{prop:last-click-not-ic})}
    & \multirow{2}{*}{(Proposition~\ref{prop:lc-is-not-fair})}
    & \makecell[l]{$n\geq3$:\ $\bigl(\,(1-\tfrac{1}{n}^{\frac{1}{n-1}})^n,\;(1 - \gamma^{2})^n\bigr]$$^*$}
    & ( Theorem~\ref{thm:last-click-n-heterogeneous})\\
    & 
    &
    & (Theorem~\ref{thm:last-click-n-homogeneous})
    &\\
  \midrule
  \multirow{4}{*}{\textbf{PVM}}
    & \multirow{2}{*}{\cmark}
    & \multirow{2}{*}{\cmark}
    & \makecell[l]{$n=2$:\quad$3/4=0.75$}
    & \makecell[l]{$n=2$:\quad$19/27\approx0.7037$}\\
    & 
    &
    & (Theorem~\ref{thm:pvm-n-homogeneous})
    &(Theorem~\ref{thm:pvm-two-heterogeneous})
    \\
    & (Proposition~\ref{prop:pvm_properties})
    & (Proposition~\ref{prop:pvm-is-fair})
    & \makecell[l]{$n\geq3$:\quad$\displaystyle1-\bigl(1-\tfrac{1}{n}\bigr)\bigl(\tfrac{1}{n}\bigr)^{\tfrac{1}{n-1}}$}
    & \makecell[l]{$n\geq3$:\quad$\bigl[(\tfrac{19}{27})^{\lceil\log_2 n\rceil},\;1-\bigl(1-\tfrac{1}{n}\bigr)\bigl(\tfrac{1}{n}\bigr)^{\tfrac{1}{n-1}}\bigr]$} \\
    &
    &
    & (Theorem~\ref{thm:pvm-n-homogeneous})
    & (Theorem~\ref{thm:pvm-n-homogeneous} \& \ref{thm:pvm-n-heterogeneous})\\
    \bottomrule
    \multicolumn{5}{l}{\footnotesize $^*$$\gamma=\sqrt[3]{\tfrac{2+\sqrt{6}}{4}}+\sqrt[3]{\tfrac{2-\sqrt{6}}{4}}$).} \\
  \end{tabular}%
}
\label{tab:pvm_settings}
\end{table}

We begin by analyzing the commonly used Last-Click Mechanism (LCM) and theoretically demonstrate that it is not DSIC. For LCM's performance, our findings reveal that, in the worst-case scenario, LCM can perform remarkably poorly. Even with just two heterogeneous platforms, both accuracy and fairness can approach arbitrarily low values.

To ensure DSIC, we propose a novel attribution mechanism called the Peer-Validated Mechanism (PVM). The mechanism operates as follows: only platforms reporting before the conversion are eligible for attribution, and the credit a platform receives depends solely on peer reports and prior probabilities---independent of its own report. Since a platform's report does not influence its own outcome, PVM is DSIC by design. We then theoretically demonstrate that PVM consistently outperforms the LCM in terms of both attribution accuracy and fairness. We further prove that it is the optimal DSIC mechanism in the homogeneous setting (Theorem~\ref{thm:pvm-optimal-homog}). Mutiple simulations using distributions fitted from real-world ad-conversion data further validate the superiority of PVM.

To the best of our knowledge,  this is the first work to formally model the advertising attribution problem within a theoretical framework, and to rigorously analyze the incentive and efficiency properties of the widely adopted Last-Click Mechanism. By shifting attention from empirical heuristics and estimation to mechanism design, our work offers foundational insights for developing attribution systems that are robust, fair, and incentive-compatible in digital advertising markets.

All missing proofs can be found in full version.

\paragraph{Related Work}
\label{sec:rela}
Recent research on advertising attribution has primarily focused on multi-touch attribution, which distributes conversion credit across multiple platforms based on observed user interactions data. A wide range of modeling approaches have been explored, including probabilistic models such as survival analysis \citep{Ji2017AdditionalMT, Ji2016Probabilistic, Shender2023TimeToEvent,Zhang2014,zhao2019revenue}, Shapley value-based methods for fair allocation \citep{Anderl2016MappingGraph-Based,  Berman2018BeyondTL, Singal2022}, and Markov models for channel transition influence \citep{Anderl2016MappingGraph-Based,Kakalejčík2018MultichannelMarketing}. Furthermore, causal inference \citep{Du2019Causally, Yao2022CausalMTA} and deep learning \citep{Ren2018DARNN, Li2018Attention, Du2019Causally, Kumar2020CAMTA, Yang2020Interpretable, Yao2022CausalMTA} have been applied to better capture temporal and interaction complexity. Despite their sophistication, these approaches generally assume that the user interactions data reported by platforms are accurate and complete.

However, this assumption often fails in practice, as platforms may strategically misreport to gain greater attribution. In contrast, mechanism design offers a principled framework for addressing strategic behavior, with incentive compatibility (IC) as a central design goal \citep{Hurwicz1972, myerson1981optimal,Vickrey1961,Clarke1971,Groves1973}. While IC-based techniques have been widely applied in domains such as auctions \citep{edelman2007internet}, voting \citep{gibbard1973manipulation}, and resource allocation \citep{rothkopf1998auction}, their application to attribution remains underexplored. Attribution presents new challenges: the strategic behavior of platforms is often ill-defined, and their utility depends on uncertain conversion outcomes, making standard mechanism design tools difficult to apply directly.

%% file: Styles/Sources/model.tex
\section{Model and Preliminaries}
\label{sec:model}
This section develops a formal model to study advertising attribution under strategic platform behavior. We first describe a typical real-world scenario, then formalize the model components, define the attribution mechanism, analyze strategic behavior, and finally define the advertiser's objective.

Throughout, we adopt the last-click  attribution standard, treating the final platform in a user's interaction sequence as the one that deserves credit. We focus on settings with at least two platforms, the minimal case where attribution ambiguity and manipulation arise.
In practice, the number of platforms involved in a conversion is typically small---often no more than five.
\subsection{Real-World Advertising Scenario}
Consider the real-world online advertising scenario where a user interacts with ads from multiple platforms ($n \ge 2$) before a conversion event. When the user converts, the advertiser seeks to allocate credit based on the click logs reported by the platforms. 

The process unfolds as follows. Each platform $i \in [n]$ first detects a user click and records a log at the corresponding \emph{absolute click time} $t_i^{\text{abs}}$. It then selects an \emph{absolute reported time} $r_i^{\text{abs}} \geq t_i^{\text{abs}}$, at which it submits the log to the advertiser. At some time $t_0 \geq \max_{i \in [n]} \{ t_i^{\text{abs}} \}$, the user converts, and the advertiser performs credit attribution based on reports received by $t_0$---that is the set $\{ r_i^{\text{abs}} \mid r_i^{\text{abs}} \leq t_0 \}$. Crucially, while the advertiser observes the conversion time $t_0$, platforms must decide when to report without knowing when the conversion will occur.

This scenario highlights the fundamental challenge in attribution: the advertiser must infer the true sequence of events based on potentially delayed reports from strategically acting platforms. The discrepancy between true click times and reported times necessitates careful mechanism design.

\subsection{Advertising Attribution Model}
We now present a game-theoretic model of the attribution process, capturing strategic platform behavior and informing mechanism design. To simplify analysis, we adopt a \emph{conversion-aligned timeline}, setting the conversion time \( t_0 = 0 \) without loss of generality. Under this transformation, all click times are expressed relative to the conversion and lie in $(-\infty, 0]$. Specifically, we define
\[
t_i := t_i^{\text{abs}} - t_0 \leq 0, \quad \forall i \in [n],
\]
where \( t_i \) denotes the relative click time of platform \( i \). Figure~\ref{fig:Conversion-aligned Time Transformation} illustrates this transformation: if two platforms record clicks at 10:40~a.m. and 10:50~a.m., and the conversion occurs at 11:00~a.m., their relative click times become \( t_1 = -20 \) and \( t_2 = -10 \), with the conversion at time 0.
\begin{figure}[h]
    \centering
    \includegraphics[width=0.65\textwidth]{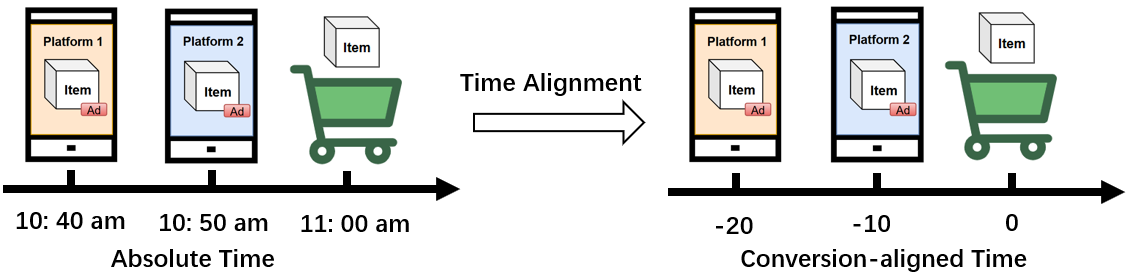}
    \caption{Conversion-aligned time transformation}
    \label{fig:Conversion-aligned Time Transformation}
\end{figure}
Under the conversion-aligned timeline, the relative click time \(\boldsymbol{t}=(t_i)_{i\in [n]} \)---which depends on the unknown conversion time---is therefore unobservable to platforms. To capture this uncertainty, we model \( t_i \) as a random variable drawn independently from a commonly known distribution, with cumulative distribution function (CDF) \( F_i(t) \) and probability density function (PDF) \( f_i(t) \), supported on \( (-\infty, 0] \). This distribution may be interpreted as a prior belief based on platform-level statistics. Let $\boldsymbol{F} = \{F_i\}_{i\in [n]}$ denote the joint distribution from which the click time vector \( \boldsymbol{t}\) is drawn.

To model strategic reporting, we assume each platform \( i \in [n] \) selects a non-negative \emph{reporting delay} \( \tau_i \geq 0 \), resulting in a reported time \( r_i = t_i + \tau_i \) on the conversion-aligned timeline\footnote{
We focus on strategic \emph{delay} in reporting rather than repeated or fraudulent submissions. 
Each platform reports its click once, possibly after a strategic delay, to maximize its attribution credit. 
This setting differs from \emph{click spamming} and \emph{click injection}, which involve multiple or fabricated reports.
}
. Since \( t_i \) is unobservable when the platform commits to its strategy, the chosen delay \( \tau_i \) is applied uniformly across all realizations of $t_i$. We denote the delay profile as \( \boldsymbol{\tau} = (\tau_i)_{i \in [n]} \), and the resulting reported time profile by \( \boldsymbol{r} = (r_i)_{i \in [n]} = \boldsymbol{t} + \boldsymbol{\tau} \). Unless stated otherwise, all subsequent analysis is conducted on the conversion-aligned timeline.

\subsection{Attribution Mechanism} 
We define an attribution mechanism $\mathcal{M}$ by its assignment rule $\mathcal{M}(\boldsymbol{r}):= \{x_i(\boldsymbol{r})\}_{i\in [n]}$, where $x_i(\boldsymbol{r})$ denotes the credit assigned to platform $i$ given the reported time profile $\boldsymbol{r}$.\footnote{In practice, reports submitted after the conversion time (i.e., $r_i>0$) are typically not received or used by the advertiser. However, for modeling generality, we allow such values as inputs to the mechanism. Their exclusion from attribution is later enforced through explicit feasibility constraints (see Constraint~(\ref{constr:post_report_ineligibility})).}
A mechanism is said to be \emph{feasible} if it satisfies the following constraints:
\begin{align}
    &  0 \leq x_i(\boldsymbol{r}) \leq 1, \quad \forall i \in [n], \boldsymbol{r} \label{constr:single_attribution_bound}\\
    &  x_i(r_i, \boldsymbol{r}_{-i}) = 0, \quad \forall i\in [n], r_i > 0, \boldsymbol{r}_{-i} \label{constr:post_report_ineligibility}\\
    &  \mathbb{E}_{\boldsymbol{t} \sim \boldsymbol{F}} \left[ \sum_{i=1}^n x_i(\boldsymbol{t} + \boldsymbol{\tau}) \right] \leq 1, \quad \forall \boldsymbol{F}, \boldsymbol{\tau}\label{constr:expected_total_attribution_constraint}
\end{align}
Constraints~\eqref{constr:single_attribution_bound} and~\eqref{constr:post_report_ineligibility} bound individual credit and exclude post-conversion reports. 
Constraint~\eqref{constr:expected_total_attribution_constraint} limits the expected credit, ensuring that the advertiser’s overall budget is respected.\footnote{
Constraint~\eqref{constr:expected_total_attribution_constraint} normalizes total credit in expectation rather than per conversion. 
This design reflects advertisers’ long-term budget control: it maintains the average expenditure over many conversions while allowing more flexibility for designing incentive-compatible mechanisms than strict per-instance normalization.
}

Given an attribution mechanism $\mathcal{M}$,  we define platform $i \in [n] $'s instantaneous utility under report profile $\boldsymbol{r}$  as the credit assigned to it by the mechanism: $u_i(\boldsymbol{r}) = x_i(\boldsymbol{r})$.    Given the distribution profile $\boldsymbol{F}$ and others' strategy $\boldsymbol{\tau}_{-i}$, platform $i$ selects its delay $\tau_i$ to maximize its expected utility:
\[
U_i(\tau_i,\boldsymbol{\tau}_{-i}) = \mathbb{E}_{\boldsymbol{t} \sim \boldsymbol{F}} \left[ u_i(\boldsymbol{t} + \boldsymbol{\tau}) \right]=\mathbb{E}_{\boldsymbol{t} \sim \boldsymbol{F}} \left[ x_i(\boldsymbol{t} + \boldsymbol{\tau}) \right].\footnote{Note that, in the absolute-time model, the platform must take expectation over the unknown conversion time and other click times. In the conversion-aligned model, the conversion time is fixed at 0, and the same uncertainty is reflected in the distribution of \( \boldsymbol{t} \).}
\]

From the advertiser's perspective, an ideal attribution mechanism $\mathcal{M}$ should achieve two primary goals: (i) incentivize truthful reporting to ensure reliable interaction data; and (ii) accurately assign credit to the platform responsible for the conversion. 

We first use \emph{dominant strategy incentive compatibility} (DSIC) to capture truthful reporting:
\begin{definition}[DSIC]
\label{def:DSIC}
A mechanism $\mathcal{M}$ is DSIC if for every platform \( i \), truthful reporting ($\tau_i=0$) maximizes its utility \( u_i(\boldsymbol{r}) \) regardless of the realized true click times \( \boldsymbol{t} = (t_i, \boldsymbol{t}_{-i}) \) or the strategies \( \boldsymbol{\tau}_{-i} \) chosen by other platforms\footnote{DSIC benefits cold-start scenarios by ensuring truthful reporting without requiring prior knowledge, enabling reliable attribution from the outset and facilitating the learning of the true distribution.}. That is, for all platforms $i \in [n]$, all true times $t_i \leq 0$ and $\boldsymbol{t}_{-i}$, all others' strategies $\boldsymbol{\tau}_{-i}$, and any deviation delay $\tau'_i > 0$:
\[
u_i(t_i, \boldsymbol{t}_{-i} + \boldsymbol{\tau}_{-i}) \geq u_i(t_i + \tau'_i, \boldsymbol{t}_{-i} + \boldsymbol{\tau}_{-i}).
\]

\end{definition}

It is easy to verify that DSIC is equivalent to a non-increasing allocation rule with respect to a platform's own report.
\begin{theorem}
\label{thm:dsic-monotonicity}
An attribution mechanism $\mathcal{M}$ satisfies DSIC if and only if, for every platform $i$ and any fixed reports from other platforms $\boldsymbol{r}_{-i}$, the credit $x_i(r_i,\boldsymbol{r}_{-i})$ is non-increasing in its own report $r_{i}$.
\end{theorem}

Second, we formalize attribution accuracy as the mechanism's ability to assign credit to the true last-click platform. Specifically, we defined the accuracy of a mechanism $\mathcal{M}$, given $\boldsymbol{F}$ as
\[
\text{ACC}(\mathcal{M}; \boldsymbol{F}) = \mathbb{E}_{\boldsymbol{t} \sim \boldsymbol{F}} \left[ \sum_{i=1}^n x_i(\boldsymbol{t} + \boldsymbol{\tau}^{\text{NE}}) \cdot \mathbb{I}[i = \arg\max_j \{t_j\}] \right],
\]

where $\boldsymbol{\tau}^{\text{NE}}$ is the Nash equilibrium induced by $\boldsymbol{F}$ and $\mathcal{M}$. Note that $\boldsymbol{\tau}^{\text{NE}} = \boldsymbol{0}$ for a DSIC mechanism.
Thus, given a known distribution \( \boldsymbol{F} \), the advertiser's optimization problem is formulated as:
\begin{equation}
\label{optimi:original_optimization}
\begin{aligned}
\max_{\mathcal{M}} \quad & ACC(\mathcal{M}; \boldsymbol{F}) \\
\text{s.t.} \quad & \text{Feasible and DSIC}.
\end{aligned}
\end{equation}

Beyond defining accuracy with respect to a fixed $\boldsymbol{F}$, we define the accuracy of mechanism $\mathcal{M}$ as
\[ACC(\mathcal{M}) := \inf_{\boldsymbol{F}} ACC(\mathcal{M}; \boldsymbol{F}),\]
capturing  worst-case performance across $\boldsymbol{F}$, serving as a evaluation for a mechanism's performance.

%% file: Styles/Sources/last-click.tex
\section{The Last-Click Mechanism}
\label{sec:last-click}

In this section, we conduct a rigorous analysis of the Last-Click Mechanism (LCM). We begin by formally defining LCM and then demonstrate that it fails to satisfy DSIC. We further evaluate its accuracy at equilibrium and derive accuracy bounds in both homogeneous and heterogeneous platform settings. Formally, the Last-Click Mechanism is defined as $\mathcal{M}_{\text{LCM}}$:
\begin{definition}[Last-Click Mechanism]
\label{def:last-click}
Given the report profile $\boldsymbol{r} = (r_i)_{i\in [n]}$, the Last-Click Mechanism is defined as $\mathcal{M}_{\text{LCM}} = \{x_i(\boldsymbol{r})\}_{i \in [n]}$. Specifically, 
\[
x_i(\boldsymbol{r}) =
\begin{cases}
1 & \text{if } i \in S \text{ and } r_i = \max_{j \in S} \{r_j\}, \\
0 & \text{otherwise,}
\end{cases}
\]
where $S = \{ j \in [n] \mid r_j \leq 0 \}$ is the set of platforms with effective reports. Ties are broken uniformly at random among the tied platforms.

\end{definition}
Due to its simplicity and its intuitive principle of crediting the platform associated with the user's final click before conversion, LCM is widely adopted in practice. However, it is easy to see that platforms may benefit from strategically delaying their reports, making truthful reporting suboptimal. Therefore, LCM does not satisfy DSIC. 

\begin{proposition}
\label{prop:last-click-not-ic}
The Last-Click Mechanism is not a DSIC mechanism.
\end{proposition}

\subsection{Accuracy Analysis}

Since LCM is not DSIC, it may assign credit to a platform that wasn't truly last, leading to inaccurate attribution.  We therefore analyze its equilibrium accuracy in both homogeneous and heterogeneous platform settings, and derive accuracy bounds for both two-platform and $n$-platform cases.\footnote{Since our focus is on DSIC mechanisms, we restrict our evaluation of LCM to instances where equilibrium is guaranteed, without analyzing its existence in general. Even within this limited scope, the results clearly demonstrate LCM's poor performance in our setting.}

We first consider the case with two homogeneous platforms and present our result in Theorem~\ref{thm:last-click-two-homogeneous}.
\begin{theorem}
\label{thm:last-click-two-homogeneous}
When there are two homogeneous platforms with identical distribution \( F(t) \), supported on \( (-\infty, 0] \), the accuracy of \( \mathcal{M_{\text{LCM}}} \) is exactly \( (2 - \sqrt{2})^2 \), and this bound is tight. 

\end{theorem}
To prove Theorem~\ref{thm:last-click-two-homogeneous}, we first analyze the incentive constraint at a symmetric strategy profile $(\tau_0, \tau_0)$. By requiring that no platform benefits by unilaterally deviating from $\tau_0$ to truthful reporting, we derive the necessary condition $F(-\tau_0) \ge 2 - \sqrt{2}$. Since LCM can only attribute correctly when both true click times are before $-\tau_0$, this yield a lower bound on accuracy of $(2 - \sqrt{2})^2$.  We then construct a family of distributions $f_M(t) = c_M (e^{-t} - 1)$, supported on $[-M, 0]$, where $c_M$ is a normalization constant. We show that this game admits a unique symmetric Nash equilibrium, and as $M \to \infty$, the accuracy converges to exactly $(2 - \sqrt{2})^2$. 

We now extend our analysis to the general case with $n$ homogeneous platforms.
\begin{theorem}
\label{thm:last-click-n-homogeneous}
When there are $n$ homogeneous platforms with identical distribution \( F(t) \), supported on \( (-\infty, 0] \), the accuracy of \( \mathcal{M_{\text{LCM}}} \) is  bounded as follows:
\[
\left(1 - \left(\frac{1}{n}\right)^{\frac{1}{n-1}}\right)^n < 
ACC(\mathcal{M_{\text{LCM}}})
\leq \left(1 - \left(\sqrt[3]{\frac{2+\sqrt{6} }{4}} + \sqrt[3]{\frac{2-\sqrt{6}}{4}}\right)^2\right)^n.
\]
\end{theorem}
The lower bound is derived using an argument similar to the two-platform case, by examining the conditions required for a symmetric equilibrium. For the upper bound, we analyze the symmetric equilibrium under a specific distribution with a linear probability density function $f(t) = -2t$ supported on $[-1, 0]$.

Finally, we consider the heterogeneous case, where each platform may follow a different distribution $F_i(t)$. Surprisingly, we show that the accuracy of the LCM can be arbitrarily low, even in a simple setting with just two heterogeneous platforms.
\begin{theorem}
\label{thm:last-click-n-heterogeneous}
When there are $n$ heterogeneous platforms with distributions $F_i(t)$, all supported on $(-\infty, 0]$. the accuracy of \( \mathcal{M_{\text{LCM}}} \) can be arbitrarily small and approach to 0.
\end{theorem}
The proof relies on a key insight: a platform with a highly concentrated distribution (e.g., supported on $(C-\epsilon, C+\epsilon)$) can easily manipulate its report to secure attribution credit. We construct an instance where one platform has such a concentrated distribution, while the others have click time supports strictly greater than it. In this scenario, we show that in equilibrium, the concentrated platform receives attribution with probability approaching 1, causing overall accuracy to approach 0. 

%% file: Styles/Sources/our_mechanism.tex
\section{The Peer-Validated Mechanism}
\label{sec:PVM}

In this section, we introduce the Peer-Validated Mechanism (PVM), a novel mechanism addressing the non-DSIC issue of LCM. Intuitively, if the credit assigned to a platform is independent of its own report, the mechanism is DSIC. Based on this idea, we propose the PVM as follows:
\begin{definition}[Peer-Validated Mechanism]
Consider \( n \) platforms with the CDF \( \{F_i\}_{i\in n} \) and PDF \( \{f_i\}_{i\in n} \) supported on \( (-\infty, 0] \). Let \( \boldsymbol{r} = (r_i)_{i \in [n]} \) be the reported time profile from \( n \) platforms. 
The \emph{Peer-Validated Mechanism} assigns credit based on mutual validation among platforms, and is defined as \( \mathcal{M}_{\text{PVM}} = \{x_i(\boldsymbol{r})\}_{i \in [n]} \), with

\[
x_i(\boldsymbol{r}) = \begin{cases}
   \mathbb{I}[r_i \leq 0]\cdot \mathbb{I}[\max_{j\in{S \setminus \{i\}}} \{r_j\} \leq \alpha_S^{(i)} ] & \text{if } |S \setminus \{i\}| \geq 1,  \\
   \mathbb{I}[r_i \leq 0] \cdot \beta_i & \text{otherwise, } \\
\end{cases}
\]

where \( S=\{j\in[n]\ | r_j \leq 0\} \) denotes platforms with eligible reports. $\beta_i = P(i = \arg\max_j \{t_j\})=\int_{-\infty}^{0} f_i(t) \prod_{j \ne i} F_j(t) \, dt$
is the probability that platform \( i \) is the true last-click platform based on the prior. The validation threshold \( \alpha_S^{(i)} \) is defined as the solution to $\prod_{j \in S \setminus \{i\}} F_j(\alpha_S^{(i)}) = \beta_i$.\footnote{The existence and uniqueness of \(\alpha_S^{(i)}\) and \(\beta_i\) under standard regularity conditions, along with handling edge cases (e.g., flat or discontinuous CDFs), are detailed in full version.}
\end{definition}
Roughly speaking, PVM assigns credit to platform $i$'s credit based on eligible reports from other platforms. When such peer reports exist, the mechanism compares them to a threshold $\alpha_S^{(i)}$, which is chosen so that the probability of all peers' true click times being no later than $\alpha_S^{(i)}$ matches the prior $\beta_i$ that platform $i$ is the true last. This validation process leverages instance-level information to make attribution decisions while preserving incentive compatibility. If no eligible peer reports are available, PVM falls back to allocating $\beta_i$ based on the prior. The reason only eligible reports are used for validation is that the mechanism assumes no overt misreporting among them, while ineligible reports ($r_j > 0$) are definitely misreports and thus excluded due to unmodeled behavior. Finally, the indicator \( \mathbb{I}[r_i \leq 0] \) ensures that attribution only goes to pre-conversion reports.

Since any reporting delay either disqualifies the platform itself or prevents others from being attributed, it is straightforward to verify that PVM satisfies feasibility and DSIC
\footnote{%
A simple variant of PVM also preserves DSIC under the click spamming problem, where a platform may repeatedly report the same click at later timestamps.
Specifically, the modified mechanism takes the first valid report (if any) from each platform as input while keeping the rest of the allocation rule unchanged.
Under this setting, the platform’s own reporting time remains decoupled from its expected number of attributions.%
}
, as formalized in Proposition~\ref{prop:pvm_properties}.
\begin{proposition}
\label{prop:pvm_properties}
The Peer-Validated Mechanism is a DSIC mechanism.  
\end{proposition}

As PVM is DSIC, we focus on truthful reports ($\boldsymbol{r} = \boldsymbol{t}$). In this case, $S = [n]$, and we let $\alpha_i$ denote the threshold used in $x_i(\cdot)$, defined by $\prod_{j \ne i} F_j(\alpha_i) = \beta_i$. The allocation rule then simplifies to
\[
x_i(\boldsymbol{t}) = \mathbb{I}[\max_{j \ne i} \{t_j\} \leq \alpha_i ]
\quad \forall i \in [n].
\]
 We adopt this reduced form throughout the remainder of our analysis of PVM.

\subsection{Optimality of PVM for Homogeneous Platforms}

We surprisingly find that PVM is the optimal DSIC mechanism in the homogeneous platform setting. 
\begin{theorem}
\label{thm:pvm-optimal-homog}
When the platforms are homogeneous, the Peer-Validated Mechanism (PVM) is the optimal DSIC mechanism with respect to the accuracy.
\end{theorem}
To show this optimality, we aim to identify the DSIC attribution rule $\{x_i(\boldsymbol{t})\}_{i\in[n]}$ that maximizes accuracy. This is a challenging task, as it involves optimizing over a set of functions simultaneously. However, if for any fixed expected attribution \( e_i = \mathbb{E}_{\boldsymbol{t}}[x_i(\boldsymbol{t})] \), we can characterize the most accurate DSIC rule that achieves it, then the problem reduces to optimizing over the expected attribution vector \( \boldsymbol{e} = (e_i)_{i\in [n]} \). The following lemma shows that such a characterization indeed exists. 
\begin{lemma}
\label{lem:threshold-best}
For platform $i$ and a fixed expected attribution $e_i$, there exists an optimal DSIC attribution rule for platform $i$ w.r.t.  accuracy, satisfying \( e_i = \mathbb{E}_{\boldsymbol{t}}[x_i(\boldsymbol{t})] \), that can be written as
\[
x_i^*(t_i, \boldsymbol{t}_{-i}) =
\begin{cases}
1, & \text{if } \max_{j \neq i} \{t_j\} \leq \theta_i, \\
0, & \text{otherwise},
\end{cases}
\]
where $G_i(t) = \Pi_{j \neq i} F_j(t)$ is the CDF of the random variable $\max_{j \neq i} \{t_j\}$, and $G_i(\theta_i)=e_i$.
\end{lemma}

Herein, we give a proof sketch of Lemma~\ref{lem:threshold-best}. First, as any DSIC rule must be non-increasing in $t_i$ (Theorem~\ref{thm:dsic-monotonicity}),  we claim that, to maximize accuracy, the optimal DSIC rule $x_i(t_i, \boldsymbol{t}_{-i})$ should be a constant when given $\boldsymbol{t}_{-i}$, so that larger values of $t_i$, which more better indicate that platform $i$ is last, are not penalized. Second, given a fixed expected attribution $e_i$, the self-independent rule $x_i(\boldsymbol{t}_{-i})$ should prioritize instances with smaller $\max_{j \ne i} \{t_j\}$, where platform $i$ is more likely to be last. This greedy strategy yields the threshold-form optimal DSIC rule in the lemma. When $G_i$ is somewhere flat, multiple thresholds may achieve $e_i$, and combining them may yield  non-threshold variants. Still, at least one such optimal rule exists.

Based on Lemma~\ref{lem:threshold-best}, the task reduces to finding the optimal  $(e^*_i)_{i\in[n]}$. Since $G_i(\theta_i)=e_i$, the original optimization problem (\ref{optimi:original_optimization}) can therefore be reformulated in terms of $\boldsymbol{\theta}=(\theta_i)_{i\in [n]}$ as follows:
\begin{align*}
\max_{\boldsymbol{\theta}}\quad &\sum_{i=1}^n \int_{-\infty}^{\theta_i} g_i(u) (1 - F_i(u)) \, du\\
 \text{s.t.}\quad &\sum_{i=1}^n G_i(\theta_i) = 1
\end{align*}
In particular, for homogeneous platforms, the thresholds defined within the PVM precisely align with the solution to the optimization problem outlined above, establishing its optimality as stated in Theorem~\ref{thm:pvm-optimal-homog}.

\subsection{Accuracy Analysis}
We now analyze the accuracy of PVM. For the homogeneous setting, we give a tight bound on the accuracy. Since PVM is the optimal DSIC mechanism, this accuracy is the maximum value a DSIC mechanism can achieve. 
\begin{theorem}
\label{thm:pvm-n-homogeneous}
When there are $n$ homogeneous platforms with identical distribution \( F(t) \), supported on \( (-\infty, 0] \), the accuracy of \( \mathcal{M}_{\text{PVM}} \) is exactly equal to
\[
\text{ACC}(\mathcal{M}_{\text{PVM}}) = 1 - \left( 1 - \frac{1}{n} \right) \left( \frac{1}{n} \right)^{\frac{1}{n-1}}.
\]
\end{theorem}
In the homogeneous case, symmetry implies that all thresholds $\alpha_i$ are equal, denoted by $\alpha$, and satisfy $F(\alpha)^{n-1} =  1/n$. Therefore, PVM makes a correct attribution if either all reports are no greater than $\alpha$, which occurs with probability $(1/n)^{n/(n-1)}$, or exactly one report exceeds $\alpha$, which occurs with probability $1 - (1/n)^{1/(n-1)}$. These probabilities depend only on $n$, not on the specific distribution. Summing them gives the accuracy in Theorem~\ref{thm:pvm-n-homogeneous}. 

In practice, the number of platforms $n$ typically does not exceed 5. We therefore conduct a comparison with the Last-Click mechanism (presented in Table \ref{tab:pvm_lc_comparison}) to show that PVM is strictly superior. 

\begin{table}[htbp] 
    \centering
    \caption{The accuracy comparison between PVM and LCM (Upper bound).}
    \begin{tabular}{cccc}
        \toprule 
        $\boldsymbol{n}$ & $\mathcal{M}_{\text{PVM}}$  & Upper bound of $\mathcal{M}_{\text{LCM}}$  &  Ratio ($\mathcal{M}_{\text{PVM}}$ / $\mathcal{M}_{\text{LCM}}$ UB) \\ 
        \midrule 
        2 & 0.75 &  $(2-\sqrt{2})^2\approx0.3431$ (tight bound)& 2.1857\\
        3 & 0.6151 & 0.3336 & 1.8437 \\ 
        4 & 0.5275 & 0.2314 & 2.2799 \\ 
        5 & 0.4650 & 0.1605 & 2.8977 \\ 
        \bottomrule 
    \end{tabular}
    \label{tab:pvm_lc_comparison} 
\end{table}

In the rest, we consider the general heterogeneous-platform setting. We show a tight bound on accuracy for two-platform case (Theorem~\ref{thm:pvm-two-heterogeneous}) and a lower bound for $n$-platform case (Theorem~\ref{thm:pvm-n-heterogeneous}).
\begin{theorem}
\label{thm:pvm-two-heterogeneous}
When there are two heterogeneous platforms, the accuracy of \( \mathcal{M}_{\text{PVM}} \) is exactly equal to $\text{ACC}(\mathcal{M}_{\text{PVM}}) = 19/27\approx 0.7037$.
\end{theorem}
To establish the result, we first formulate an optimization problem that characterizes the worst-case accuracy by maximizing the misattribution probability. In the two-platform setting, all attribution outcomes can be explicitly enumerated, making this optimization analytically tractable. To show tightness, we then construct a concrete instance that satisfies the optimality conditions, thereby achieving the accuracy value of 
$19/27$. 
\begin{theorem}
\label{thm:pvm-n-heterogeneous}
When there are $n$ heterogeneous platforms, the lower bound on the accuracy of \( \mathcal{M}_{\text{PVM}} \) is  $\text{ACC}(\mathcal{M}_{\text{PVM}}) = (19/27)^{\lceil\log_2 n\rceil}$.
\end{theorem}
For $n$ heterogeneous platforms, we design a binary-tree-based mechanism to derive a lower bound for PVM. Starting from the root node, which represents all $n$ platforms, we recursively partition them into two disjoint subsets $L$ and $R$ of sizes $\lceil n/2 \rceil$ and $\lfloor n/2 \rfloor$, respectively. Each subset is treated as a virtual platform, represented by the distribution of $\max_{i \in L} \{t_i\}$ or $\max_{i \in R} \{t_i\}$. At each internal node, the attribution reduces to a problem between two heterogeneous platforms. Repeating this over $\lceil \log_2 n \rceil$ levels yields an overall accuracy lower bound of $(19/27)^{\lceil \log_2 n \rceil}$. Since PVM is guaranteed to perform at least as well as this mechanism, the same expression serves as a lower bound for its accuracy. 

\subsection{Fairness of PVM}
Besides the DSIC and accuracy, 
PVM also satisfies a strong \emph{fairness} property: 
the expected attribution $\mathbb{E}[x_i]$ for each platform $i$ exactly matches its true probability of contributing the last click, $P(i = \arg\max_j \{t_j\})$. This alignment offers a principled basis for evaluating long-term platform effectiveness and simultaneously promotes trust in the mechanism's equity.
To quantify this alignment and enable comparisons across mechanisms, we define the following metric:

\begin{definition}
\label{def:fairness-metric}
The \emph{fairness score} of mechanism \( \mathcal{M} \) under the joint distribution \( \boldsymbol{F} \) is defined as
\[
FAIR(\mathcal{M}; \boldsymbol{F}) = \min_{\{i \mid P(i = \arg\max_j \{t_j\}) > 0\}} \left\{ \frac{\mathbb{E}[x_i(\boldsymbol{t})]}{P(i = \arg\max_j \{t_j\})} \right\}.
\]
\end{definition}

\begin{definition}
\label{def:fairness-property}
A mechanism \( \mathcal{M} \) is \emph{Fair} if, for any joint distribution \( \boldsymbol{F} \), it holds that
\[
FAIR(\mathcal{M}; \boldsymbol{F}) = 1.
\]
\end{definition}
 The fairness score $FAIR(\mathcal{M}; \boldsymbol{F})$  quantifies how closely a mechanism's expected attribution matches the true last-click probabilities, with a score of 1 indicates perfect alignment. A \emph{Fair} mechanism ensures that attribution faithfully reflects contribution probabilities across all distributions. PVM is a fair mechanism directly from the choice of $(\alpha_i)_{i \in [n]}$ under DSIC:
 \[
 \mathbb{E}_{\boldsymbol{t}}[x_i(\boldsymbol{t})] = \mathbb{E}_{\boldsymbol{t}}[\max_{j\neq i}\{t_j\}\leq \alpha_i] = \prod_{j\neq i}F_j(\alpha_i)= P(i = \arg\max_j \{t_j\}).
 \]

\begin{proposition}
    \label{prop:pvm-is-fair} 
    The Peer-Validated Mechanism is \emph{Fair}.
\end{proposition}

In contrast, the Last-Click Mechanism fails to meet this property.

\begin{proposition}
\label{prop:lc-is-not-fair}
The Last-Click Mechanism is not \emph{Fair}.
\end{proposition}

LCM fails the \emph{Fair} property due to its fairness score being highly sensitive to distributional differences and strategic delays, especially under heterogeneity. As shown in Table~\ref{tab:lc-fairness-bounds}, the fairness score can degrade to zero in such settings.

\begin{table}[htbp] 
\centering
\caption{Worst-Case Fairness Score of LCM under Equilibrium.}
\label{tab:lc-fairness-bounds}
\begin{tabular}{cc}
\toprule
Scenario & Worst-Case Fairness ($\inf_{\boldsymbol{F}} FAIR(\mathcal{M}_{LCM}; \boldsymbol{F})$) \\
\hline
Homogeneous, $n=2$ & $ 1 - (\sqrt{2}-1)^2 \approx 0.828$ \\
Homogeneous, $n \ge 3$ & $ (1 - (1/n)^{n/(n-1)},1-(\sqrt[3]{\frac{2+\sqrt{6}}{4}}+\sqrt[3]{\frac{2-\sqrt{6}}{4}})^{2n}]$ \\ 
Heterogeneous, $n \ge 2$ & $0$ \\
\bottomrule
\end{tabular}
\end{table}

%% file: Styles/Sources/experiment.tex
\section{Numerical Experiments}
\label{sec:simulations}
We empirically evaluate PVM against LCM using simulations based on click time distributions fitted from real-world ad conversion logs from four advertising platforms. Experiments cover two settings: homogeneous and heterogeneous.  In the homogeneous case, we simulate
$n\in\{2,3,4,5\}$ identical platforms, all following the click time same distribution, repeated across four distributions derived from real data.  In the heterogeneous case with $n=2$, we simulate all six platform pairs formed by different combinations of the four distributions. Under LCM, platforms play in equilibrium; under PVM, they report truthfully by DSIC. Each configuration was evaluated using $5 \times 10^4$ simulated user paths, repeated over 10 independent runs.

PVM consistently outperforms LCM in both accuracy and fairness across all settings. Table~\ref{tab:overall_improvement_summary} reports the improvements as mean $\pm$ standard deviation over platforms (homogeneous) or platform pairs (heterogeneous). Specifically, accuracy gains grew with $n$ (up to $0.3041$ when $n=5$) and remains notable under heterogeneity ($0.0655$). Fairness improvements are small in homogeneous cases but substantial in heterogeneous ones ($0.1320$).
\begin{table}[htbp]
  \centering
  \small
    \caption{Aggregate summary of PVM's improvements over LCM, (mean \(\pm\) standard deviation)}
  \resizebox{\textwidth}{!}{%
  \begin{tabular}{lccccc}
    \toprule
     & \multicolumn{4}{c}{Homo Setting} & Hetero Setting \\
    \cmidrule(lr){2-5}\cmidrule(lr){6-6}
    Metric & \(n=2\) & \(n=3\) & \(n=4\) & \(n=5\) & (over 6 pairs) \\
    \midrule
    Acc.    & \(0.0404\pm0.0396\) & \(0.1583\pm0.0439\) & \(0.2444\pm0.0580\) & \(0.3041\pm0.0490\) & \(0.0655\pm0.0283\) \\
    Fair. & \(0.0248\pm0.0089\) & \(0.0157\pm0.0034\) & \(0.0107\pm0.0047\) & \(0.0111\pm0.0040\) & \(0.1320\pm0.0598\) \\
    \bottomrule
  \end{tabular}%
  }
  \label{tab:overall_improvement_summary}
\end{table}

%% file: Styles/Sources/conclusion.tex
\vspace*{-15pt}
\section{Conclusion and Discussion}
\label{sec:conclusion}
This paper introduces a formal game-theoretic framework for advertising attribution under strategic platform behavior. We show that the widely used Last-Click Mechanism fails to be dominant strategy incentive-compatible (DSIC) and performs poorly in both accuracy and fairness. To address these limitations, we propose the Peer-Validated Mechanism (PVM), a novel DSIC mechanism that allocates credit based on peer reports. We prove that PVM achieves optimal accuracy in homogeneous settings, offers provable guarantees in heterogeneous ones, and satisfies a strong fairness property. Our theoretical analysis is further validated by numerical experiments using real-world data, where PVM consistently outperforms LCM.

In practice, peer-validation principle offers a concrete design guideline for incentive-compatible attribution systems. For instance, in machine learning-based models, excluding a platform's own report as an input feature ensures truthfulness, shifting the focus from detection to design. 

Besides, PVM framework can be extended to settings with correlated click-time distributions while preserving the core peer-validation principle and the DSIC property. The validation rule generalizes from a scalar threshold~($\alpha_i$) to a multi-dimensional acceptance region~$D_i$ over peer reports~$t_{-i}$, constructed greedily by including outcomes with the highest posterior probability that platform~$i$ was the true last click until $P(t_{-i}\!\in\!D_i)=\beta_i$. A platform receives credit if and only if its peers’ reports fall within~$D_i$. Under this modification, the homogeneous-case results remain unchanged, since our proofs for those theorems do not rely on the independence assumption; the results for the last-click mechanism also remain the same, as its accuracy and fairness are already zero; while in heterogeneous settings, PVM retains a weaker but still meaningful $1/n$ lower bound on accuracy. We focus on the independence assumption in this paper to present the mechanism’s core insight in the clearest setting, which is sufficient to capture the essential strategic structure, leaving correlated extensions for future work.

Several directions remain open. First, while PVM aligns the expected attribution with true last-click probability, future work may explore mechanisms that further improve instance-level accuracy. Second, investigating correlated click-time distributions could enhance a mechanism's applicability in realistic scenarios. Next, a joint optimization framework modeling both advertiser and platform utilities, integrating attribution with budget allocation, represents a compelling direction. Finally, investigating repeated games with externalities—where platforms may strategically harm peers or misreport distributions to manipulate learned priors—could address dynamic interactions, potentially incorporating bidding strategies for a more comprehensive ecosystem model.

%% file: Styles/Sources/app.tex
\section{Missing Proofs in Section~\ref{sec:model}}  
\label{sec:appA}  
\subsection{Proof of Theorem~\ref{thm:dsic-monotonicity}: DSIC Monotonicity}  
\label{app:proof-DSIC-monotonicity}  
\begin{proof}  
Let $\mathcal{M} = (x_i(\boldsymbol{r}))_{i=1}^n$ be a feasible attribution mechanism, with utility $u_i(\boldsymbol{r}) = x_i(\boldsymbol{r})$ for each platform $i$. By Definition~\ref{def:DSIC}, $\mathcal{M}$ is DSIC if for all $i$, any true click time $t_i \le 0$, any reports $\boldsymbol{r}_{-i}$, and any $\tau'_i > 0$, it holds that
\[
x_i(t_i, \boldsymbol{r}_{-i}) \ge x_i(t_i + \tau'_i, \boldsymbol{r}_{-i}).
\]

We prove the equivalence between DSIC and monotonicity.

\paragraph{Step 1: DSIC $\Rightarrow$ Monotonicity.}  
Fix any $r_a < r_b$. Consider three cases:

\begin{itemize}
    \item \textbf{Case 1: $r_b \le 0$}. Let $t_i = r_a$, $\tau'_i = r_b - r_a > 0$. By DSIC,
    \[
    x_i(r_a, \boldsymbol{r}_{-i}) \ge x_i(r_b, \boldsymbol{r}_{-i}).
    \]
    
    \item \textbf{Case 2: $r_a \le 0 < r_b$}. By feasibility (Constraint~\eqref{constr:post_report_ineligibility}), $x_i(r_b, \boldsymbol{r}_{-i}) = 0$, and by non-negativity, $x_i(r_a, \boldsymbol{r}_{-i}) \ge 0$. Hence,
    \[
    x_i(r_a, \boldsymbol{r}_{-i}) \ge x_i(r_b, \boldsymbol{r}_{-i}).
    \]
    
    \item \textbf{Case 3: $0 < r_a < r_b$}. Both are post-conversion reports, so feasibility implies
    \[
    x_i(r_a, \boldsymbol{r}_{-i}) = x_i(r_b, \boldsymbol{r}_{-i}) = 0.
    \]
\end{itemize}

In all cases, $x_i(r_i, \boldsymbol{r}_{-i})$ is non-increasing in $r_i$.

\paragraph{Step 2: Monotonicity $\Rightarrow$ DSIC.}  
Assume $x_i(r_i, \boldsymbol{r}_{-i})$ is non-increasing in $r_i$. Let $t_i \le 0$ be the true click time, and $\tau'_i \ge 0$ a deviation. Let $r_a = t_i$, $r_b = t_i + \tau'_i$. Then $r_a \le r_b$, and by monotonicity,
\[
x_i(t_i, \boldsymbol{r}_{-i}) \ge x_i(t_i + \tau'_i, \boldsymbol{r}_{-i}),
\]
which satisfies the DSIC condition.
\end{proof}

\section{Missing Proofs in Section \ref{sec:last-click}}
\label{sec:appB}
\subsection{Proof of Proposition \ref{prop:last-click-not-ic}}
\label{app:proof-prop-last-click-not-ic}
\begin{proof}
To prove that the LCM  is not DSIC, we demonstrate a scenario where the necessary monotonicity condition from Theorem~\ref{thm:dsic-monotonicity} is violated. Specifically, we show an instance where for a platform $i$ and fixed reports $\boldsymbol{r}_{-i}$ from others, the credit $x_i(r_i, \boldsymbol{r}_{-i})$ is increasing in its own report $r_i$.

Consider a simple case with $n=2$ platforms.
Let the report of platform 2 be fixed at $r_2 = -10$. We examine  the allocation function $x_1(r_1,-10)$:
\begin{itemize}
    \item \textbf{Case 1}: Platform 1 reports $r_a = -20$.
The report vector is $\boldsymbol{r} = (r_a, r_2) = (-20, -10)$. The set of eligible reports is $S = \{1, 2\}$. The maximum eligible report is $\max\{-20, -10\} = -10$. Since $r_a = -20 \neq \max_{j \in S}\{r_j\}$, platform 1 receives zero attribution: $x_1(-20, -10) = 0$.
\item \textbf{Case 2}: Platform 1 reports $r_b = -5$.
The report vector is $\boldsymbol{r}' = (r_b, r_2) = (-5, -10)$. The set of eligible reports is $S' = \{1, 2\}$. The maximum eligible report is $\max\{-5, -10\} = -5$. Since $r_b = -5 = \max_{j \in S'}\{r_j\}$, platform 1 receives $x_1(-5, -10) = 1$.

\end{itemize}

Comparing the two cases, we have reports $r_a = -20$ and $r_b = -5$ such that $r_a < r_b \le 0$. However, the attribution for platform 1 increases: $x_1(r_a, -10) = 0 < 1 = x_1(r_b, -10)$.

This violates the condition that $x_i(r_i, \boldsymbol{r}_{-i})$ must be non-increasing in $r_i$ for a mechanism to be DSIC (Theorem~\ref{thm:dsic-monotonicity}). Therefore, the Last-Click Mechanism is not DSIC.
\end{proof}

\subsection{Proof of Theorem \ref{thm:last-click-two-homogeneous}}
\label{app:proof-last-click-thm1}
\begin{proof}
Let $f(t)$ be the common PDF of the true display times $t_1$ and $t_2$, defined on $(-\infty, 0]$, with CDF $F(t)$. We analyze the Last-Click Mechanism under a symmetric Nash Equilibrium strategy profile $(\tau_0, \tau_0)$. The reported times are $r_1 = t_1 + \tau_0$ and $r_2 = t_2 + \tau_0$.

\paragraph{Step 1: Lower Bound on Accuracy}
Since $(\tau_0, \tau_0)$ is a symmetric Nash Equilibrium, neither platform has an incentive to unilaterally deviate. Consider platform 2 (WLOG). Its expected utility $U_2(\tau_2; \tau_1 = \tau_0)$ must be maximized at $\tau_2 = \tau_0$.

The expected utility for platform 2 when playing $\tau_0$ (given platform 1 plays $\tau_0$) is calculated based on the probability it receives attribution. By symmetry, when both platforms submit valid reports ($r_1 \le 0, r_2 \le 0$), each has a $1/2$ chance of winning credit. The probability that at least one report is valid is $1 - \mathbb{P}(r_1 > 0, r_2 > 0) = 1 - \mathbb{P}(t_1 > -\tau_0, t_2 > -\tau_0) = 1 - (1 - F(-\tau_0))^2$. Therefore, the expected utility for platform 2 is:
\[
U_2(\tau_0;\tau_1= \tau_0) = \frac{1}{2} \left[ 1 - (1 - F(-\tau_0))^2 \right]
\]

Now, consider the utility if platform 2 deviates to truthful reporting, $\tau_2 = 0$. Its reported time is $r_2 = t_2$. Platform 2 gets attribution if $r_2 = t_2 \le 0$ (which is always true for $t_2 \in (-\infty, 0]$) and either $r_1 = t_1 + \tau_0 > 0$ (platform 1's report is invalid) or  $r_2 = t_2 \ge r_1 = t_1 + \tau_0$).
The utility $U_2(0; \tau_1 =\tau_0)$ is:
\begin{align*} U_2(0; \tau_1=\tau_0) &= \mathbb{P}(t_1 + \tau_0 > 0) + \mathbb{P}(t_2 \ge t_1 + \tau_0) \\ &\ge \mathbb{P}(t_1 + \tau_0 > 0) \\ &= \mathbb{P}(t_1 > -\tau_0) = 1 - F(-\tau_0) \end{align*}
Since $(\tau_0, \tau_0)$ is a Nash Equilibrium, the utility from playing $\tau_0$ must be at least as high as the utility from deviating to $0$: $U_2(\tau_0; \tau_1=\tau_0) \ge U_2(0;\tau_1= \tau_0)$.
\[
\frac{1}{2} \left[ 1 - (1 - F(-\tau_0))^2 \right] \ge 1 - F(-\tau_0)
\]
Let $x = F(-\tau_0)$. Since $t \in (-\infty, 0]$, $F(-\tau_0) \in [0, 1]$.
\[
\frac{1}{2} [1 - (1 - x)^2] \ge 1 - x
\]
\[
1 - (1 - 2x + x^2) \ge 2 - 2x
\]
\[
2x - x^2 \ge 2 - 2x
\]
\[
x^2 - 4x + 2 \le 0
\]
The roots of $x^2 - 4x + 2 = 0$ are $x = \frac{4 \pm \sqrt{16 - 8}}{2} = 2 \pm \sqrt{2}$. The inequality $x^2 - 4x + 2 \le 0$ holds for $x \in [2 - \sqrt{2}, 2 + \sqrt{2}]$. Since $x = F(-\tau_0) \in [0, 1]$, we have $x \in [2 - \sqrt{2},1]$.
Thus, in any symmetric NE $(\tau_0, \tau_0)$, $F(-\tau_0) \ge 2 - \sqrt{2}$.

The accuracy is the probability of correct attribution. A necessary condition for correct attribution is that both reports are valid ($r_1 \le 0, r_2 \le 0$), which occurs with probability $F(-\tau_0)^2$. Therefore, in two homogeneous setting, the accuracy in any symmetric NE is bounded below by:
\[
ACC(\mathcal{M}_{\text{LCM}})=\inf_{F}ACC(\mathcal{M_{\text{LCM}}};F^{\times2}) \ge F(-\tau_0)^2 \ge (2 - \sqrt{2})^2
\]

\paragraph{Step 2: Tightness of the Bound}
We construct an example using the distribution family with PDF $f_M(t) = c_M (e^{-t} - 1)$ for $t \in [-M, 0]$ and $0$ otherwise, where $M > 0$ and $c_M = (e^M - M - 1)^{-1}$ is the normalization constant. The CDF is $F_M(t) = c_M (-e^{-t} - t + e^M - M)$ for $t \in [-M, 0]$. We aim to show that for this family, a unique symmetric Nash equilibrium exists, and the accuracy approaches $(2 - \sqrt{2})^2$ as $M \to \infty$.

\textit{Existence and Uniqueness of Best Response:}
First, we demonstrate that for any fixed strategy $\tau_2 \in [0, M]$ chosen by platform 2, there exists a unique best response strategy $\tau_1^* \in (0, M)$ for platform 1. Let $U_1(\tau_1; \tau_2)$ denote the expected utility of platform 1. The best response $\tau_1^*$ maximizes this utility. We analyze the first and second derivatives of $U_1$ with respect to $\tau_1$ (using the notation $e'$ and $e''$ from the sketch):
\[
e'(\tau_1|\tau_2) = \frac{\partial U_1(\tau_1; \tau_2)}{\partial \tau_1} = -f_M(-\tau_1) + \int_{-M}^{-\tau_1} f_M(t) f_M(t + \tau_1 - \tau_2) dt
\]
\[
e''(\tau_1|\tau_2) = \frac{\partial^2 U_1(\tau_1; \tau_2)}{\partial \tau_1^2} = f_M'(-\tau_1) - f_M(-\tau_1)f_M(-\tau_2) + \int_{-M}^{-\tau_1} f_M(t) f_M'(t + \tau_1 - \tau_2) dt
\]
Since $f_M'(t) = -c_M e^{-t} < 0$ and $f_M(t) \ge 0$ for $t \in [-M, 0]$, we established that $e''(\tau_1|\tau_2) < 0$ for $\tau_1 \in [0, M)$. This proves that $U_1(\tau_1; \tau_2)$ is a strictly concave function of $\tau_1$ for any fixed $\tau_2$.

Consequently, the first derivative $e'(\tau_1|\tau_2)$ is strictly decreasing in $\tau_1$. We examine its sign at the boundaries:
\begin{itemize}
    \item $e'(0|\tau_2) = \int_{-M}^{0} f_M(t) f_M(t - \tau_2) dt > 0$.
    \item $e'(M|\tau_2) = -f_M(-M) = -c_M(e^M - 1) < 0$ (for $M>0$).
\end{itemize}
Since $e'(\tau_1|\tau_2)$ is continuous (as $f_M$ and $f_M'$ are continuous) and strictly decreases from a positive value at $\tau_1=0$ to a negative value at $\tau_1=M$, the Intermediate Value Theorem guarantees that there exists \textit{exactly one} value $\tau_1^* \in (0, M)$ such that $e'(\tau_1^*|\tau_2) = 0$. This unique root $\tau_1^*$ corresponds precisely to the unique maximum of the strictly concave utility function $U_1(\tau_1; \tau_2)$. Therefore, for any given strategy $\tau_2$, there exists a unique best response $\tau_1^* \in (0, M)$ for platform 1, characterized by the unique solution to the first-order condition $e'(\tau_1|\tau_2) = 0$. By symmetry, the same holds for platform 2.

\textit{Finding the Symmetric Nash Equilibrium:}
Having established the uniqueness of the best response, we now seek a symmetric equilibrium $(\tau_M, \tau_M)$. This requires $\tau_M$ to be the best response to itself, satisfying $e'(\tau_M|\tau_M) = 0$:
\[
e'(\tau_M|\tau_M) = -f_M(-\tau_M) + \int_{-M}^{-\tau_M} f_M(t)^2 dt = 0
\]
The analysis confirms this equation has a unique solution $\tau_M \in (0, M)$. We solve:
\[
-c_M(e^{\tau_M} - 1) + c_M^2 \int_{-M}^{-\tau_M} (e^{-t} - 1)^2 dt = 0
\]
Multiplying by $(e^M - M - 1)^2$:
\[
-(e^M - M - 1)(e^{\tau_M} - 1) + \int_{-M}^{-\tau_M} (e^{-2t} - 2e^{-t} + 1) dt = 0
\]
\[
-(e^M - M - 1)(e^{\tau_M} - 1) + \left[ -\frac{1}{2}e^{-2t} + 2e^{-t} + t \right]_{-M}^{-\tau_M} = 0
\]
\[
-(e^M - M - 1)(e^{\tau_M} - 1) + \left( -\frac{1}{2}e^{2\tau_M} + 2e^{\tau_M} - \tau_M \right) - \left( -\frac{1}{2}e^{2M} + 2e^{M} - M \right) = 0
\]
Rearranging terms yields the equation for $\tau_M$:
\[
-\frac{1}{2}e^{2\tau_M} + (3 + M - e^M)e^{\tau_M} - \tau_M + \frac{1}{2}e^{2M} - e^M - 1 = 0
\]
Let $\tau_M = M + \gamma$, where $\gamma \in (-M, 0)$. Substituting and dividing by $e^{2M}$:
\[
-\frac{1}{2}e^{2\gamma} + (3+M)e^{-M}e^{\gamma} - e^{\gamma} - (M+\gamma)e^{-2M} + \frac{1}{2} - e^{-M} - e^{-2M} = 0
\]
Taking the limit as $M \to \infty$:
\[
-\frac{1}{2}e^{2\gamma} - e^{\gamma} + \frac{1}{2} = 0 \implies e^{2\gamma} + 2e^{\gamma} - 1 = 0
\]
Solving the quadratic equation for $e^\gamma$ yields $e^\gamma = \frac{-2 \pm \sqrt{4 - 4(1)(-1)}}{2} = -1 \pm \sqrt{2}$. Since $e^\gamma > 0$, we must have $e^\gamma = \sqrt{2} - 1$.

\textit{Calculating Limiting Accuracy:}
The accuracy $\mathcal{A}_M$ in the equilibrium $(\tau_M, \tau_M)$ is at least $F_M(-\tau_M)^2$. We evaluate the limit of this lower bound:
\begin{align*} \lim_{M \to \infty} F_M(-\tau_M) &= \lim_{M \to \infty} F_M(-M - \gamma) \\ &= \lim_{M \to \infty} c_M (-e^{M+\gamma} - (-M-\gamma) + e^M - M) \\ &= \lim_{M \to \infty} \frac{-e^M e^\gamma + \gamma + e^M}{e^M + M - 1} \\ &= \lim_{M \to \infty} \frac{e^M(1 - e^\gamma) + \gamma}{e^M+M - 1} \\ &= 1 - e^\gamma \quad (\text{since } e^M \text{ dominates } M \text{ as } M \to \infty) \\ &= 1 - (\sqrt{2} - 1) = 2 - \sqrt{2} \end{align*}
Therefore, in two homogeneous setting, $ ACC(\mathcal{M}_{LCM}) \leq (\lim_{M \to \infty} F_M(-\tau_M))^2 = (2 - \sqrt{2})^2$.
Since we established a general lower bound of $(2 - \sqrt{2})^2$ for the accuracy in any symmetric NE, and we constructed a specific family of distributions for which the accuracy approaches this value in the limit, we conclude that the worst-case accuracy is exactly $(2 - \sqrt{2})^2$, and the bound is tight.
\end{proof}

\subsection{Proof of Theorem \ref{thm:last-click-n-homogeneous}}
\label{app:proof-last-click-thm-n-homo}
\begin{proof}
 Let $f(t)$ and $F(t)$ denote the common PDF and CDF of the true impression times $t_i$ for $n$ homogeneous platforms, supported on $(-\infty, 0]$. In a symmetric Nash equilibrium, all platforms adopt the same strategy $(\tau_0, \ldots, \tau_0)$, so the reported times are $r_i = t_i + \tau_0$. The LCM attributes credit to the platform with the latest valid report ($r_i \leq 0$). The accuracy $ACC(\mathcal{M}_{\text{LCM}};F^{\times n})$ is the probability that the platform with the latest true impression time $t_i < 0$ is correctly attributed, which requires all reports to be valid ($r_i \leq 0$), occurring with probability $F(-\tau_0)^n$.

We prove the theorem in two parts: first, we establish that for any $F$, in any symmetric equilibrium, 
\[ACC(\mathcal{M}_{\text{LCM}};F^{\times n}) > \left(1 - \left(\frac{1}{n}\right)^{\frac{1}{n-1}}\right)^n;
\] second, we construct a distribution $F_0$ with a symmetric equilibrium such that 
\[ ACC(\mathcal{M}_{\text{LCM}};F_0^{\times n}) \leq \left(1 - \left(\left(\frac{2+\sqrt{6}}{4}\right)^{\frac{1}{3}} + \left(\frac{2-\sqrt{6}}{4}\right)^{\frac{1}{3}}\right)^2\right)^n.\]

\paragraph{Step 1: Lower Bound on Accuracy}
Assume $(\tau_0, \ldots, \tau_0)$ is a symmetric Nash equilibrium, where no platform increases its expected utility by unilateral deviation. The expected utility $U_i(\tau_i; \boldsymbol{\tau}_{-i} =[\tau_0]^{n-1})$ of platform $i$ choosing strategy $\tau_i$, with others choosing $\tau_0$, is the probability it receives attribution.

When all platforms choose $\tau_0$, the probability that at least one report is valid ($r_j \leq 0$) is:
\begin{align*}
1 - P(r_j > 0, \forall j) &= 1 - P(t_j > -\tau_0, \forall j) \\
&= 1 - (1 - F(-\tau_0))^n.
\end{align*}
By symmetry, each platform's utility is:
\[
U_i(\tau_i; \boldsymbol{\tau}_{-i} =[\tau_0]^{n-1}) = \frac{1}{n} \left[ 1 - (1 - F(-\tau_0))^n \right].
\]

If platform $i$ deviates to $\tau_i = 0$, its report $r_i = t_i \leq 0$ is valid, and it is attributed if all other reports are invalid ($t_j > -\tau_0$ for $j \neq i$). This occurs with probability:
\[
U_i(0; \boldsymbol{\tau}_{-i} =[\tau_0]^{n-1}) \geq (1 - F(-\tau_0))^{n-1}.
\]

The equilibrium condition requires:
\[
\frac{1}{n} \left[ 1 - (1 - F(-\tau_0))^n \right] \geq (1 - F(-\tau_0))^{n-1}.
\]
Multiplying by $n$:
\[
1 - (1 - F(-\tau_0))^n \geq n (1 - F(-\tau_0))^{n-1}.
\]

Let $x = 1 - F(-\tau_0) \in [0, 1]$. The inequality becomes:
\[
1 - x^n \geq n x^{n-1} \quad \text{or} \quad g(x) = x^n + n x^{n-1} - 1 \leq 0.
\]
Compute the derivative:
\[
g'(x) = n x^{n-1} + n(n-1) x^{n-2} = n x^{n-2} (x + n-1) > 0 \quad \text{for } x \in [0,1].
\]
Thus, $g(x)$ is strictly increasing. Evaluate at:
\begin{itemize}
    \item $x = \left( \frac{1}{n+1} \right)^{\frac{1}{n-1}}$: $g(x) < 0$ (since $\left( \frac{1}{n+1} \right)^{\frac{1}{n-1}} < 1$).
    \item $x = \left( \frac{1}{n} \right)^{\frac{1}{n-1}}$: $g(x) = \left( \frac{1}{n} \right)^{\frac{n}{n-1}} + n \cdot \frac{1}{n} - 1 > 0$.
\end{itemize}

Since $g(x)$ is strictly increasing, there exists a unique root $x_0 \in \left( \left( \frac{1}{n+1} \right)^{\frac{1}{n-1}}, \left( \frac{1}{n} \right)^{\frac{1}{n-1}} \right)$ such that $g(x_0) = 0$. Thus, $g(x) \leq 0$ implies:
\[
x < \left( \frac{1}{n} \right)^{\frac{1}{n-1}}.
\]
Therefore:
\[
1 - F(-\tau_0) < \left( \frac{1}{n} \right)^{\frac{1}{n-1}}, \quad F(-\tau_0) > 1 - \left( \frac{1}{n} \right)^{\frac{1}{n-1}}.
\]
The accuracy in n homogeneous setting satisfies:
\[
ACC(\mathcal{M}_{\text{LCM}}) = \inf_{F} ACC(\mathcal{M}_{\text{LCM}};F^{\times n})= F(-\tau_0)^n > \left(1 - \left( \frac{1}{n} \right)^{\frac{1}{n-1}}\right)^n.
\]

\paragraph{Step 2: Upper Bound on Accuracy}
We construct a distribution to demonstrate that the accuracy can approach the upper bound in a symmetric equilibrium. Define the PDF:
\[
f(t) = \begin{cases} 
-2t & \text{if } t \in [-1, 0], \\ 
0 & \text{otherwise}.
\end{cases}
\]
Verify that $f(t)$ is a valid PDF:
\[
\int_{-1}^{0} f(t) \, dt = \int_{-1}^{0} (-2t) \, dt = -2 \left[ \frac{t^2}{2} \right]_{-1}^{0} = -2 \left( 0 - \frac{1}{2} \right) = 1.
\]
The CDF is:
\[
F(t) = \int_{-1}^{t} (-2u) \, du = -\left[ u^2 \right]_{-1}^{t} = -(t^2 - 1) = 1 - t^2, \quad t \in [-1, 0].
\]
Compute derivatives:
\[
f'(t) = -2, \quad f''(t) = 0.
\]

Part 2 is divided into two subparts: first, we prove the existence and uniqueness of platform $i$'s best response when other platforms' strategies are fixed at $\boldsymbol{\tau}_{-i} = [\tau_2]^{\times n}$; second, we establish the symmetric equilibrium and compute the upper bound on accuracy.

\textbf{\textit{Existence and Uniqueness of Best Response:}}
Fix the strategies of platforms $j \neq i$ at $\boldsymbol{\tau}_{-i} = [\tau_2]^{\times n}$. Platform $i$'s expected utility is:
\[
U_i(\tau_i; \boldsymbol{\tau}_{-i} = [\tau_2]^{\times n}) = \int_{-1}^{-\tau_i} f(t) \left( 1 - \int_{t + \tau_i - \tau_2}^{-\tau_2} f(u) \, du \right)^{n-1} \, dt.
\]

The first derivative is:
\begin{align*}
&U_i'(\tau_i; \boldsymbol{\tau}_{-i} = [\tau_2]^{\times n}) \\
=& -f(-\tau_i) + \int_{-1}^{-\tau_i} f(t) (n-1) \left( 1 - \int_{t + \tau_i - \tau_2}^{-\tau_2} f(u) \, du \right)^{n-2} f(t + \tau_i - \tau_2) \, dt \\
=& -f(-\tau_i) + \int_{-1}^{-\tau_i} f(t) \left( \left( 1 - \int_{t + \tau_i - \tau_2}^{-\tau_2} f(u) \, du \right)^{n-1} \right)' \, dt \\
=& -f(-\tau_i) + \left[ f(t) \left( 1 - \int_{t + \tau_i - \tau_2}^{-\tau_2} f(u) \, du \right)^{n-1} \right]_{-1}^{-\tau_i} - \int_{-1}^{-\tau_i} f'(t) \left( 1 - \int_{t + \tau_i - \tau_2}^{-\tau_2} f(u) \, du \right)^{n-1} \, dt \\
=& -f(-1) \left( 1 - \int_{-1 + \tau_i - \tau_2}^{-\tau_2} f(u) \, du \right)^{n-1} - \int_{-1}^{-\tau_i} f'(t) \left( 1 - \int_{t + \tau_i - \tau_2}^{-\tau_2} f(u) \, du \right)^{n-1} \, dt.
\end{align*}

The second derivative is:
\begin{align*}
&U_i''(\tau_i; \boldsymbol{\tau}_{-i} = [\tau_2]^{\times n}) \\
=& -f(-1) (n-1) \left( 1 - \int_{-1 + \tau_i - \tau_2}^{-\tau_2} f(u) \, du \right)^{n-2} f(-1 + \tau_i - \tau_2) + f'(-\tau_i) \\
&\quad - \int_{-1}^{-\tau_i} f'(t) (n-1) \left( 1 - \int_{t + \tau_i - \tau_2}^{-\tau_2} f(u) \, du \right)^{n-2} f(t + \tau_i - \tau_2) \, dt \\
\leq &-f(-1) (n-1) \left( 1 - \int_{-1 + \tau_i - \tau_2}^{-\tau_2} f(u) \, du \right)^{n-2} f(-1 + \tau_i - \tau_2) + f'(-\tau_i) \\
&\quad - \left[ f'(t) \left( 1 - \int_{t + \tau_i - \tau_2}^{-\tau_2} f(u) \, du \right)^{n-1} \right]_{-1}^{-\tau_i} + \int_{-1}^{-\tau_i} f''(t) \left( 1 - \int_{t + \tau_i - \tau_2}^{-\tau_2} f(u) \, du \right)^{n-1} \, dt \\
=& -f(-1) (n-1) \left( 1 - \int_{-1 + \tau_i - \tau_2}^{-\tau_2} f(u) \, du \right)^{n-2} f(-1 + \tau_i - \tau_2) \\
&\quad + f'(-1) \left( 1 - \int_{-1 + \tau_i - \tau_2}^{-\tau_2} f(u) \, du \right)^{n-1} < 0.
\end{align*}

Thus, $U_i'(\tau_i; \tau_{-i} = \tau_2)$ is strictly decreasing. Since:
\[
U_i'(0; \tau_{-i} = \tau_2) = \int_{-1}^{0} f(t) (n-1) \left( 1 - \int_{t - \tau_2}^{-\tau_2} f(u) \, du \right)^{n-2} f(t - \tau_2) \, dt > 0,
\]
\[
U_i'(1; \tau_{-i} = \tau_2) = -f(-1) < 0,
\]
there exists a unique $\tau_i^* \in (0, 1)$ such that $U_i'(\tau_i^*; \boldsymbol{\tau}_{-i} = [\tau_2]^{\times n}) = 0$, which is the best response.

\textbf{\textit{Symmetric Equilibrium:}}

For the symmetric equilibrium, set $\tau_i = \tau$, $\boldsymbol{\tau}_{-i} = [\tau]^{\times n}$. The first-order condition is:
\begin{align*}
&U_i'(\tau; \boldsymbol{\tau}_{-i} = [\tau]^{\times n}) = -f(-\tau) + (n-1) \int_{-1}^{-\tau} f(t) \left( 1 - \int_{t}^{-\tau} f(u) \, du \right)^{n-2} f(t) \, dt \\
=& -f(-1) \left( 1 - \int_{-1}^{-\tau} f(u) \, du \right)^{n-1} - \int_{-1}^{-\tau} f'(t) \left( 1 - \int_{t}^{-\tau} f(u) \, du \right)^{n-1} \, dt.
\end{align*}
Define $h(\tau) = U_i'(\tau; \boldsymbol{\tau}_{-i} = [\tau]^{\times n})$. Compute its derivative:
\begin{align*}
&h'(\tau) = -(n-1) f(-1) \left( 1 - \int_{-1}^{-\tau} f(u) \, du \right)^{n-2} f(-\tau) + f'(-\tau) \\
&\quad  - (n-1) \int_{-1}^{-\tau} f'(t) \left( 1 - \int_{t}^{-\tau} f(u) \, du \right)^{n-2} f(-\tau) \, dt \\
&\leq -(n-1) f(-1) \left( 1 - \int_{-1}^{-\tau} f(u) \, du \right)^{n-2} f(-\tau) + f'(-\tau) \\
&\quad - (n-1) \int_{-1}^{-\tau} f'(t) \left( 1 - \int_{t}^{-\tau} f(u) \, du \right)^{n-2} f(t) \, dt \\
&= -(n-1) f(-1) \left( 1 - \int_{-1}^{-\tau} f(u) \, du \right)^{n-2} f(-\tau) + f'(-\tau) \\
&\quad - \left[ f'(t) \left( 1 - \int_{t}^{-\tau} f(u) \, du \right)^{n-1} \right]_{-1}^{-\tau} + \int_{-1}^{-\tau} f''(t) \left( 1 - \int_{t}^{-\tau} f(u) \, du \right)^{n-1} \, dt \\
&= -(n-1) f(-1) \left( 1 - \int_{-1}^{-\tau} f(u) \, du \right)^{n-2} f(-\tau) \\
&\quad + f'(-1) \left( 1 - \int_{-1}^{-\tau} f(u) \, du \right)^{n-1} < 0.
\end{align*}
Thus, $h(\tau)$ is strictly decreasing. Since:
\[
h(0) = (n-1) \int_{-1}^{0} f(t) \left( 1 - \int_{t}^{0} f(u) \, du \right)^{n-2} f(t) \, dt > 0,
\]
\[
h(1) = -f(-1) < 0,
\]
there exists a unique $\tau_0 \in (0, 1)$ such that $h(\tau_0) = 0$, defining the symmetric equilibrium.

For $f(t) = -2t$, $t \in [-1, 0]$, otherwise $0$:
\begin{align*}
h(\tau) &= -2 \tau^{2n-2} + 2 \int_{\tau}^{1} (1 + \tau^2 - u^2)^{n-1} \, du.
\end{align*}
Set $h(\tau) = 0$:
\[
\tau^{2n-2} = \int_{\tau}^{1} (1 + \tau^2 - u^2)^{n-1} \, du.
\]
For $n = 2$:
\begin{align*}
\tau^2 &= \int_{\tau}^{1} (1 + \tau^2 - u^2) \, du = (1 - \tau) + \left[ \tau^2 u - \frac{u^3}{3} \right]_{\tau}^{1} \\
&= (1 - \tau) + \left( \tau^2 - \frac{1}{3} \right) - \left( \tau^3 - \frac{\tau^3}{3} \right) = \tau^2 + \frac{2}{3} - \tau - \frac{2}{3} \tau^3.
\end{align*}
Thus:
\[
2 - 3\tau - 2\tau^3 = 0.
\]
The real root is:
\[
\tau_0 = \left( \frac{2+\sqrt{6}}{4} \right)^{\frac{1}{3}} +\left( \frac{2-\sqrt{6}}{4} \right)^{\frac{1}{3}}.
\]
Then:
\[
F(-\tau_0) = 1 - \tau_0^2 = 1 - \left( \left( \frac{2+\sqrt{6}}{4} \right)^{\frac{1}{3}} + \left( \frac{2-\sqrt{6}}{4} \right)^{\frac{1}{3}} \right)^2.
\]
Denote the specific distribution as $F_0$, and the accuracy is:
\[
ACC(\mathcal{M}_{\text{LCM}};F_0^{\times 2} ) = F_0(-\tau_0)^2 = \left( 1 - \left( \left( \frac{2+\sqrt{6}}{4} \right)^{\frac{1}{3}} + \left( \frac{2-\sqrt{6} }{4} \right)^{\frac{1}{3}} \right)^2 \right)^2.
\]
For $n > 2$, let $A(u) = 1 + \tau^2 - u^2 > 0$, $u \in [\tau, 1]$. By Jensen's inequality for $x^{n-1}$ ($n-1 \geq 1$):
\[
\left( \frac{1}{1 - \tau} \int_{\tau}^{1} A(u) \, du \right)^{n-1} \leq \frac{1}{1 - \tau} \int_{\tau}^{1} A(u)^{n-1} \, du.
\]
Thus:
\[
\left( \int_{\tau}^{1} A(u) \, du \right)^{n-1} \leq \int_{\tau}^{1} A(u)^{n-1} \, du.
\]
So:
\begin{align*}
h(\tau) &= 2 \left[ -\tau^{2n-2} + \int_{\tau}^{1} (1 + \tau^2 - u^2)^{n-1} \, du \right] \\
&\geq 2 \left[ -\tau^{2n-2} + \left( \int_{\tau}^{1} (1 + \tau^2 - u^2) \, du \right)^{n-1} \right].
\end{align*}
Since:
\[
-\tau_0^{2n-2} + \left( \int_{\tau_0}^{1} (1 + \tau_0^2 - u^2) \, du \right)^{n-1} = 0,
\]
we have $h(\tau_0) \geq 0$. As $h'(\tau) < 0$, the root $\tau_n$ for $n > 2$ satisfies $\tau_n \geq \tau_0$. Thus:
\[
F(-\tau_n) \leq F(-\tau_0).
\]
The worst-case accuracy satisfies:
\begin{align*}
ACC(\mathcal{M}_{\text{LCM}})=&\inf_{F}ACC(\mathcal{M}_{\text{LCM}};F^{\times n})  \leq F(-\tau_n)^n \leq F(-\tau_0)^n \\
=& \left( 1 - \left( \left( \frac{2+\sqrt{6}}{4} \right)^{\frac{1}{3}} + \left( \frac{2-\sqrt{6}}{4} \right)^{\frac{1}{3}} \right)^2 \right)^n.
\end{align*}

\textbf{\textit{Conclusion}}

Combining both parts, the $ACC(\mathcal{M}_{\text{LCM}})$ in n homogeneous setting satisfies:
\[
\left(1 - \left( \frac{1}{n} \right)^{\frac{1}{n-1}}\right)^n < ACC(\mathcal{M}_{\text{LCM}}) \leq \left(1 - \left( \left( \frac{2+\sqrt{6}}{4} \right)^{\frac{1}{3}} + \left( \frac{2-\sqrt{6}}{4} \right)^{\frac{1}{3}} \right)^2 \right)^n.
\]
\end{proof}

\subsection{Proof of Theorem \ref{thm:last-click-n-heterogeneous}}
\label{app:proof-last-click-n-heterogeneous}
\begin{proof}
To prove that the worst-case accuracy of the Last-Click Mechanism for $n$ platforms with distributions $f_i(t)$ over $(-\infty, 0]$ can approach 0 in equilibrium, and that this bound is tight, we construct a specific set of impression time distributions. We show that a Nash equilibrium exists where a platform that is never the true last impression receives nearly all attribution, driving the accuracy to 0. The proof proceeds in five steps: defining the distributions, introducing an equivalent distribution, establishing the equilibrium, computing the attribution, and analyzing the accuracy in the original setting.

\textbf{Step 1: True Impression Time Distributions}
Consider $n$ platforms with impression times $t_i$, $i \in [n]$. The true PDFs are defined as follows:

\begin{itemize}
    \item \textit{Platform 1}: The impression time $t_1$ follows a uniform distribution over a narrow interval:
    \[
    f_{1,\text{true}}(t) = \begin{cases}
    \frac{1}{\epsilon} & \text{if } t \in [-2-\epsilon, -2], \\
    0 & \text{otherwise},
    \end{cases}
    \]
    where $\epsilon>0$ is a small parameter. The integral $\int_{-2-\epsilon}^{-2} \frac{1}{\epsilon} \, dt = 1$ confirms that $f_{1,\text{true}}(t)$ is a valid PDF.

    \item \textit{Platforms $i \in \{2, \ldots, n\}$}: The impression times $t_i$ are all uniformly distributed:
    \[
    f_i(t) = \begin{cases}
    1 & \text{if } t \in [-1, 0], \\
    0 & \text{otherwise}.
    \end{cases}
    \]
    The integral $\int_{-1}^{0} 1 \, dt = 1$ verifies that each $f_i(t)$ is a valid PDF.
\end{itemize}

For any realizations $t_1 \sim f_{1,\text{true}}$ and $t_i \sim f_i$, we have $t_1 \in [-2-\epsilon, -2]$ and $t_i \in [-1, 0]$. Since $-t_1 \leq -2 < -1 \leq t_i$, platform 1 is never the true last impression. The true last impression is always one of platforms $2, \ldots, n$, each equally likely due to their identical distributions. 

\textbf{Step 2: Equivalent Equilibrium-Inducing Distribution}

To simplify equilibrium analysis, we define a shifted distribution for platform 1:
\[
f_{1,\text{eq}}(t) = \begin{cases}
\frac{1}{\epsilon} & \text{if } t \in [-\epsilon, 0], \\
0 & \text{otherwise}.
\end{cases}
\]

This distribution is related to $f_{1,\text{true}}$ by a shift: $f_{1,\text{eq}}(t) = f_{1,\text{true}}(t + 2)$. Specifically, $f_{1,\text{true}}(t) = \frac{1}{\epsilon}$ for $t \in [-2-\epsilon, -2]$, and shifting the argument by 2 units (i.e., $t \mapsto t + 2$) maps the interval $[-2-\epsilon, -2]$ to $[-\epsilon, 0]$. The distributions $f_i(t)$ for $i \in \{2, \ldots, n\}$ remain unchanged, i.e., $f_i(t) = 1$ for $t \in [-1, 0]$.

The shift aligns platform 1's distribution support with that of the other platforms, facilitating the analysis of Nash equilibrium under $f_{1,\text{eq}}$ and $(f_i)_{i\in \{2\cdots n\}}$. Suppose a Nash equilibrium exists where platform 1 reports $t_1 + \tau_1$ under $f_{1,\text{eq}}$. In the original setting with $f_{1,\text{true}}$ and $(f_i)_{i=2}^n$, platform 1 can report $t_1 + 2 + \tau_1$, compensating for the 2-unit shift between $f_{1,\text{true}}$ and $f_{1,\text{eq}}$. Since the Last-Click Mechanism depends on the relative ordering of reported times, this adjustment ensures that the attribution probabilities remain identical. Therefore, we analyze the equilibrium using $f_{1,\text{eq}}$ and verify the results in the original setting.

\textbf{Step 3: Existence of Nash Equilibrium}

We hypothesize a equilibrium where platforms $i \in \{2, \ldots, n\}$ all report $\tau_2$ , and platform 1 reports optimally ($\tau_1$). The Last-Click Mechanism attributes the click to the platform with the latest reported time $t_i + \tau_i$.
\textit{Platforms $i \in \{2, \ldots, n\}$}: For any platform $i \in \{2, \ldots, n\}$, given platform 1's strategy $\tau_1$ and the common strategy $\tau_2$ of platforms $j \in \{2, \ldots, n\} \setminus \{i\}$, the expected attribution of platform $i$ as a function of its strategy $\tau$ is:
\[
e_i(\tau | \tau_1, \tau_2) = \int_{-\infty}^{-\tau} f_2(t) \left(1 - \int_{t+\tau-\tau_2}^{-\tau_2} f_2(u) \, du \right)^{n-2} \left(1 - \int_{t+\tau-\tau_1}^{-\tau_1} f_{1,\text{eq}}(u) \, du \right) dt.
\]
The derivative is:
\begin{align*}
&e_i'(\tau | \tau_1, \tau_2) = -f_2(-\tau) + \int_{-\infty}^{-\tau} f_2(t) \left[ (n-2) \left(1 - \int_{t+\tau-\tau_2}^{-\tau_2} f_2(u) \, du \right)^{n-3} f_2(t+\tau-\tau_2) \right. \\
&\quad \left. \times \left(1 - \int_{t+\tau-\tau_1}^{-\tau_1} f_{1,\text{eq}}(u) \, du \right) + \left(1 - \int_{t+\tau-\tau_2}^{-\tau_2} f_2(u) \, du \right)^{n-2} f_{1,\text{eq}}(t+\tau-\tau_1) \right] dt \\
&= -f_2(-\tau) + \int_{-1}^{-\tau} f_2(t) \left[ \left(1 - \int_{t+\tau-\tau_2}^{-\tau_2} f_2(u) \, du \right)^{n-2} \left(1 - \int_{t+\tau-\tau_1}^{-\tau_1} f_{1,\text{eq}}(u) \, du \right) \right]' dt \\
&= -f_2(-\tau) + \left[ f_2(t) \left(1 - \int_{t+\tau-\tau_2}^{-\tau_2} f_2(u) \, du \right)^{n-2} \left(1 - \int_{t+\tau-\tau_1}^{-\tau_1} f_{1,\text{eq}}(u) \, du \right) \right]_{-1}^{-\tau} \\
&\quad - \int_{-1}^{-\tau} f_2'(t) \left(1 - \int_{t+\tau-\tau_2}^{-\tau_2} f_2(u) \, du \right)^{n-2} \left(1 - \int_{t+\tau-\tau_1}^{-\tau_1} f_{1,\text{eq}}(u) \, du \right) dt \\
&= -f_2(-1) \left(1 - \int_{-1+\tau-\tau_2}^{-\tau_2} f_2(u) \, du \right)^{n-2} \left(1 - \int_{-1+\tau-\tau_1}^{-\tau_1} f_{1,\text{eq}}(u) \, du \right) < 0,
\end{align*}
since $f_2'(t) = 0$ in $[-1, 0]$, $f_2(-1) = 1$, and the boundary term at $t = -\tau$ cancels out. Thus, regardless of $\tau_1$ and $\tau_2$, platform $i$'s optimal strategy is to report truthfully ($\tau = 0$).

By symmetry among platforms $i \in \{2, \ldots, n\}$, if an equilibrium exists where these platforms adopt identical strategies, then each platform $i \in \{2, \ldots, n\}$ reports truthfully.

\textit{Platform 1}: Now consider platform 1 when platforms $2, \ldots, n$ report truthfully ($\tau_i = 0$). The expected attribution function is:
\[
e_1(\tau) = \int_{-\epsilon}^{-\tau} f_{1,\text{eq}}(t) \left(1 - \int_{t+\tau}^{0} f_2(u) \, du \right)^{n-1} dt.
\]
The derivative is:
\begin{align*}
e_1'(\tau) &= -f_{1,\text{eq}}(-\tau) + \int_{-\epsilon}^{-\tau} f_{1,\text{eq}}(t) (n-1) \left(1 - \int_{t+\tau}^{0} f_2(u) \, du \right)^{n-2} f_2(t+\tau) \, dt \\
&= -f_{1,\text{eq}}(-\tau) + \int_{-\epsilon}^{-\tau} f_{1,\text{eq}}(t) \left[ \left(1 - \int_{t+\tau}^{0} f_2(u) \, du \right)^{n-1} \right]' dt \\
&= -f_{1,\text{eq}}(-\tau) + \left[ f_{1,\text{eq}}(t) \left(1 - \int_{t+\tau}^{0} f_2(u) \, du \right)^{n-1} \right]_{-\epsilon}^{-\tau} \\
&= -f_{1,\text{eq}}(-\tau) + f_{1,\text{eq}}(-\tau) - f_{1,\text{eq}}(-\epsilon) \left(1 - \int_{-\epsilon+\tau}^{0} f_2(u) \, du \right)^{n-1} \\
&= -f_{1,\text{eq}}(-\epsilon) \left(1 - \int_{-\epsilon+\tau}^{0} f_2(u) \, du \right)^{n-1} < 0.
\end{align*}
Thus, when platform 1 follows $f_{1,\text{eq}}$ and platforms $2, \ldots, n$ report truthfully, platform 1's optimal strategy is to report truthfully ($\tau = 0$).

\textbf{Step 4: Attribution Calculation at Equilibrium}

Thus, truthful reporting is an equilibrium under distributions $f_{1,\text{eq}}$ and $f_2$. In this equilibrium, where all platforms report truthfully ($\tau_i = 0$), platform 1's attribution is:
\[
e_1(0) = \int_{-\epsilon}^{0} \frac{1}{\epsilon} \left(1 - \int_{t}^{0} 1 \, du \right)^{n-1} dt.
\]
Substitute $s = -t$, so $dt = -ds$, and the limits become $s = 0$ to $s = \epsilon$:
\[
e_1(0) = \int_{0}^{\epsilon} \frac{1}{\epsilon} (1 - s)^{n-1} \, ds = \frac{1}{\epsilon} \left[ -\frac{1}{n} (1 - s)^n \right]_{0}^{\epsilon} = \frac{1}{n\epsilon} \left[ 1 - (1 - \epsilon)^n \right].
\]
When $\epsilon \to 0$:
\[
\lim_{\epsilon \to 0} \frac{1}{n} \frac{1}{\epsilon} \left[ 1 - (1 - \epsilon)^n \right] = \frac{1}{n} \lim_{\epsilon \to 0} \frac{1 - (1 - \epsilon)^n}{\epsilon} = \frac{1}{n} \lim_{\epsilon \to 0} \frac{n (1 - \epsilon)^{n-1}}{1} = 1.
\]
Thus, as $\epsilon \to 0$, platform 1's attribution $e_1(0) \to 1$, and the attribution for platforms $i \in \{2, \ldots, n\}$ is:
\[
e_i(0) = \frac{1 - e_1(0)}{n-1} \to 0.
\]

\textbf{Step 5: Accuracy in the Original Setting and Tightness}

In the original setting with $f_{1,\text{true}}$ for platform 1 and $f_2(t)$ for the other platforms, platform 1 reports $t_1 + 2$ to mimic the equilibrium strategy under $f_{1,\text{eq}}$. Since $f_{1,\text{eq}}(t) = f_{1,\text{true}}(t + 2)$, this ensures the same attribution probabilities, with $e_1(0) \to 1$ and $\sum_{i=2}^{n} e_i(0) \to 0$ as $\epsilon \to 0$. Although platform 1 is never the true last impression, the true last impression is always one of platforms $2, \ldots, n$. The accuracy, defined as the probability that the attributed platform is the true last impression, satisfies:
\[
ACC(\mathcal{M}_{\text{LCM}}) \leq ACC(\mathcal{M}_{\text{LCM}};\boldsymbol{F}) \leq \sum_{i=2}^{n} e_i(0) \to 0,
\]
since the probability of correctly attributing the true last impression is at most the total attribution probability of platforms $2, \ldots, n$, which approaches 0. This demonstrates that the worst-case accuracy approaches 0. The construction shows the bound is tight, as the accuracy can be made arbitrarily close to 0 by choosing sufficiently small $\epsilon$.
\end{proof}

\section{Missing Proofs in Section~\ref{sec:PVM}}
\label{sec:appC}
\subsection{Existence and Uniqueness of Thresholds}
\label{app:threshold-existence}

Here we provide a detailed analysis of the existence and uniqueness of the validation threshold \(\alpha_S^{(i)}\) and the prior probability \(\beta_i\) in the Peer-Validated Mechanism (PVM), as referenced in the main text. Here, \(\alpha_S^{(i)}\) is the threshold for platform \(i\) given the set \(S = \{j \in [n] \mid r_j \leq 0\}\), solving \(G_i(\alpha_S^{(i)}) = \beta_i\), where \(G_i(t) = \prod_{j \in S \setminus \{i\}} F_j(t)\) is the CDF of the maximum click time among \(i\)'s peers, and \(\beta_i = \int_{-\infty}^{0} f_i(t) \prod_{j \ne i} F_j(t) \, dt\) is the prior probability that \(i\) is the true last-click platform.

\textbf{Standard Case: }
If \(G_i(t)\) is continuous and strictly increasing, then in this standard case, for any \(\beta_i \in (0,1)\), a unique solution \(\alpha_S^{(i)}\) is guaranteed to exist by the Intermediate Value Theorem.

\textbf{Special Case 1 (Flat CDF): }
If \(G_i(t)\) has a flat region, and \(\beta_i\) falls within this flat range, the threshold \(\alpha_S^{(i)}\) would not be unique. However, a flat region in \(G_i(t)\) implies that the probability of the maximum peer click time occurring within that specific interval is zero (\(P(\max{\boldsymbol{t}_{-i}} \in [a,b]) = 0\)). Therefore, any threshold \(\alpha_S^{(i)}\) chosen within this flat interval will yield the exact same attribution outcome. The choice is arbitrary and has no impact on the mechanism's performance.

\textbf{Special Case 2 (Discontinuous CDF): }
In the rare event that \(G_i(t)\) could have a jump discontinuity at a point \(\theta\) such that \(G_i(\theta-) < \beta_i < G_i(\theta)\), a single threshold \(\alpha_S^{(i)}\) would not exist. This scenario implies a non-zero probability that \(\max{\boldsymbol{t}_{-i}} = \theta\). The mechanism can be modified to handle this by using a probabilistic assignment: if \(\max{\boldsymbol{t}_{-i}} = \theta\), we assign credit to platform \(i\) with a specific probability \(p\) such that the expected attribution remains \(\beta_i\). However, we consider this case to be of limited practical relevance, as click times are typically modeled as continuous random variables.

\subsection{Proof of Proposition \ref{prop:pvm_properties}}
\label{app:proof_pvm_properties}

\begin{proof}
We prove the two properties separately.
\paragraph{Part 1: PVM is Feasible}
We verify the three feasibility constraints defined in Section~\ref{sec:model}.

\textit{Constraint 1: $0 \le x_i(\boldsymbol{r}) \le 1$ for all $i \in [n]$ and all $\boldsymbol{r}$.}
It is easy to check that $x_i(\boldsymbol{r})\in \{0,e_i,1\}$. Therefore, it satisfied Constraint 1. 

\textit{Constraint 2: $x_i(\boldsymbol{r}) = 0$ if $r_i > 0$.}

The term $\mathbb{I}[r_i \leq 0]$ is part of the definition of $x_i(\boldsymbol{r})$. If $r_i > 0$, then $\mathbb{I}[r_i \leq 0] = 0$, which makes the entire product $x_i(\boldsymbol{r}) = 0$. This constraint is satisfied.

\textit{Constraint (3): $\mathbb{E}_{\boldsymbol{t} \sim \boldsymbol{F}}[\sum_{i=1}^n x_i(\boldsymbol{t} +\boldsymbol{\tau})] \le 1 \quad \forall\boldsymbol{F}, \boldsymbol{\tau}$ }

Suppose there are $n$ platforms with click time distributions $\boldsymbol{F}$, and let $\boldsymbol{\tau}$ denote the reporting strategy profile. Let $e_i \in [0,1]$ denote the true probability that platform $i$ is the last click, and $\hat{e}_i$ its expected attribution under PVM. Define $\mathbb{S}$ as the set of all subsets of $[n]$, and for each subset $S$, let $P(S)$ be the probability that exactly the platforms in $S$ report before the conversion time (i.e., are eligible). Then:

$$
\begin{aligned}
\hat{e}_i &= \sum_{S \in \mathbb{S}} P(S)\, \mathbb{I}[i \in S] \left( \mathbb{I}[|S \setminus \{i\}| \ge 1]\, P\left( \max_{j \in S \setminus \{i\}} r_j \le \alpha_S^{(i)} \right) + \mathbb{I}[|S \setminus \{i\}| = 0]\, e_i \right) \\
&\le \sum_{S \in \mathbb{S}} P(S)\, \mathbb{I}[i \in S] \left( \mathbb{I}[|S \setminus \{i\}| \ge 1]\, P\left( \max_{j \in S \setminus \{i\}} t_j \le \alpha_S^{(i)} \right) + \mathbb{I}[|S \setminus \{i\}| = 0]\, e_i \right) \\
&= \sum_{S \in \mathbb{S}} P(S)\, \mathbb{I}[i \in S] \left( \mathbb{I}[|S \setminus \{i\}| \ge 1]\, e_i + \mathbb{I}[|S \setminus \{i\}| = 0]\, e_i \right) \\
&= \sum_{S \in \mathbb{S}} P(S)\, \mathbb{I}[i \in S]\, e_i \\
&\le e_i.
\end{aligned}
$$

Summing over all $i \in [n]$, we obtain:

$$
\mathbb{E}_{\boldsymbol{t} \sim \boldsymbol{F}} \left[ \sum_{i=1}^n x_i(\boldsymbol{t} + \boldsymbol{\tau}) \right] = \sum_{i=1}^n \hat{e}_i \le \sum_{i=1}^n e_i = 1,
$$

which confirms that Constraint 3 is satisfied.
\paragraph{Part 2: PVM is DSIC}
Fix any $\boldsymbol{r}_{-i}$. Under the PVM, the allocation rule $x_i(r_i, \boldsymbol{r}_{-i})$ depends on $r_i$ only through the indicator $\mathbb{I}[r_i \leq 0]$, which is a non-increasing function of $r_i$. Therefore, $x_i(r_i, \boldsymbol{r}_{-i})$ is also non-increasing in $r_i$.

By Theorem~\ref{thm:dsic-monotonicity}, this monotonicity implies that PVM satisfies the DSIC property. 
\end{proof}

\subsection{Proof of Lemma \ref{lem:threshold-best}}
\label{app:proof-threshold-best}
\begin{proof}
We aim to maximize the expected correct attribution, defined as
\[
    \mathbb{E}_{\boldsymbol{t}} \left[ x_i(t_i, \boldsymbol{t}_{-i}) \cdot \mathbb{I}[t_i > \max\{\boldsymbol{t}_{-i}\}] \right],
\]
subject to two constraints: (1) $x_i(t_i, \boldsymbol{t}_{-i})$ must be non-negative and non-increasing in $t_i$ for any fixed $\boldsymbol{t}_{-i}$ (ensuring DSIC), and (2) the expected attribution is fixed, i.e., $\mathbb{E}_{\boldsymbol{t}}[x_i(t_i, \boldsymbol{t}_{-i})] = e_i$. We also require that $x_i(t_i, \boldsymbol{t}_{-i}) \in [0,1]$ for all $t_i$, reflecting the feasibility.

To determine the optimal allocation, we first derive an upper bound on the conditional accuracy for a given \( \max\{\boldsymbol{t}_{-i}\} = m \in (-\infty, 0] \). Fix \(\boldsymbol{t}_{-i}^{(0)} \) such that \( \max\{\boldsymbol{t}_{-i}^{(0)}\} = m \). The expected attribution is:
\[
\mathbb{E}_{t_i}[x_i(t_i, \boldsymbol{t}_{-i}^{(0)})] = \int_{-\infty}^0 x_i(t_i, \boldsymbol{t}_{-i}^{(0)}) f_i(t_i) \, dt_i,
\]
and the expected correct attribution, occurring when \( t_i > m \), is:
\[
\mathbb{E}_{t_i}[x_i(t_i, \boldsymbol{t}_{-i}^{(0)}) \cdot \mathbb{I}[t_i > m]  ] = \int_m^0 x_i(t_i, \boldsymbol{t}_{-i}^{(0)}) f_i(t_i) \, dt_i.
\]
If \( \mathbb{E}_{t_i}[x_i(t_i, \boldsymbol{t}_{-i}^{(0)})] = 0 \), no attribution is assigned, and the accuracy is undefined. If \( \int_m^0 x_i(t_i, \boldsymbol{t}_{-i}^{(0)}) f_i(t_i) \, dt_i = 0 \), the correct attribution is zero, so the accuracy is zero, which does not contribute to deriving a non-trivial upper bound. Thus, we assume \( \mathbb{E}_{t_i}[x_i(t_i, \boldsymbol{t}_{-i}^{(0)})] > 0 \) and \( \int_m^0 x_i(t_i, \boldsymbol{t}_{-i}^{(0)}) f_i(t_i) \, dt_i > 0 \), ensuring both the denominator and numerator are non-zero. The conditional accuracy, weighted by the density \( g_i(\boldsymbol{t}_{-i}^{(0)}) \), is:
\[
\frac{\mathbb{E}_{t_i}[x_i(t_i, \boldsymbol{t}_{-i}^{(0)}) \cdot \mathbb{I}[t_i > m] ]}{\mathbb{E}_{t_i}[x_i(t_i, \boldsymbol{t}_{-i}^{(0)})} = \frac{\int_m^0 x_i(t_i, \boldsymbol{t}_{-i}^{(0)}) f_i(t_i) \, dt_i}{\int_{-\infty}^0 x_i(t_i, \boldsymbol{t}_{-i}^{(0)}) f_i(t_i) \, dt_i}.
\]
Splitting the denominator, we obtain:
\[
\int_{-\infty}^0 x_i(t_i, \boldsymbol{t}_{-i}^{(0)}) f_i(t_i) \, dt_i = \int_{-\infty}^m x_i(t_i, \boldsymbol{t}_{-i}^{(0)}) f_i(t_i) \, dt_i + \int_m^0 x_i(t_i, \boldsymbol{t}_{-i}^{(0)}) f_i(t_i) \, dt_i,
\]
so the accuracy becomes:
\[
\frac{\int_m^0 x_i(t_i, \boldsymbol{t}_{-i}^{(0)}) f_i(t_i) \, dt_i}{\int_{-\infty}^m x_i(t_i, \boldsymbol{t}_{-i}^{(0)}) f_i(t_i) \, dt_i + \int_m^0 x_i(t_i, \boldsymbol{t}_{-i}^{(0)}) f_i(t_i) \, dt_i} = \frac{1}{\frac{\int_{-\infty}^m x_i(t_i, \boldsymbol{t}_{-i}^{(0)}) f_i(t_i) \, dt_i}{\int_m^0 x_i(t_i, \boldsymbol{t}_{-i}^{(0)}) f_i(t_i) \, dt_i} + 1}.
\]
Since \( x_i(t_i, \boldsymbol{t}_{-i}^{(0)}) \) is non-increasing, \( x_i(t_i, \boldsymbol{t}_{-i}^{(0)}) \geq x_i(m, \boldsymbol{t}_{-i}^{(0)}) \) for \( t_i \leq m \), and \( x_i(t_i, \boldsymbol{t}_{-i}^{(0)}) \leq x_i(m, \boldsymbol{t}_{-i}^{(0)}) \) for \( t_i \geq m \). Thus:
\[
\frac{\int_{-\infty}^m x_i(t_i, \boldsymbol{t}_{-i}^{(0)}) f_i(t_i) \, dt_i}{\int_m^0 x_i(t_i, \boldsymbol{t}_{-i}^{(0)}) f_i(t_i) \, dt_i} \geq \frac{\int_{-\infty}^m x_i(m, \boldsymbol{t}_{-i}^{(0)}) f_i(t_i) \, dt_i}{\int_m^0 x_i(m, \boldsymbol{t}_{-i}^{(0)}) f_i(t_i) \, dt_i} = \frac{x_i(m, \boldsymbol{t}_{-i}^{(0)}) F_i(m)}{x_i(m, \boldsymbol{t}_{-i}^{(0)}) (1 - F_i(m))} = \frac{F_i(m)}{1 - F_i(m)}.
\]
Hence:
\[
\text{Accuracy} \leq \frac{1}{\frac{F_i(m)}{1 - F_i(m)} + 1} = \frac{1 - F_i(m)}{F_i(m) + 1 - F_i(m)} = 1 - F_i(m).
\]
Equality holds when \( x_i(t_i, \boldsymbol{t}_{-i}^{(0)}) = c > 0 \), as:
\[
\frac{\int_m^0 c f_i(t_i) \, dt_i}{\int_{-\infty}^0 c f_i(t_i) \, dt_i} = 1 - F_i(m).
\]
Thus, the conditional accuracy is at most \( 1 - F_i(m) \).

We now prove that the proposed rule:
\[
x_i^{*}(t_i, \boldsymbol{t}_{-i}) = 
\begin{cases} 
1, & \text{if } \max\{\boldsymbol{t}_{-i}\} \leq \theta_i, \\
0, & \text{otherwise},
\end{cases}
\]
with \( G_i(\theta_i) = e_i \), is optimal. This rule is DSIC, as it is independent of \( t_i \), and satisfies:
\[
\mathbb{E}_t[x_i^{*}(t_i, \boldsymbol{t}_{-i})] = \mathbb{P}(\max\{\boldsymbol{t}_{-i}\} \leq \theta_i) = G_i(\theta_i) = e_i.
\]
Suppose there exists an optimal mechanism \( x_i(t_i, \boldsymbol{t}_{-i}) \neq x_i^{*}(t_i, \boldsymbol{t}_{-i}) \), which maximizes \( \mathbb{E}_t[x_i(t_i, \boldsymbol{t}_{-i}) \cdot \mathbb{I}[t_i > \max\{\boldsymbol{t}_{-i}\}]] \), satisfies DSIC, and meets \( \mathbb{E}_t[x_i(t_i, \boldsymbol{t}_{-i})] = e_i \). Since \( x_i \neq x_i^{*} \), there exists some \( \boldsymbol{t}_{-i}^{(1)} \) with \( \max\{\boldsymbol{t}_{-i}^{(1)}\} = m_1 \leq \theta_i \) such that \( x_i(\cdot, \boldsymbol{t}_{-i}^{(1)}) \) is not identically 1, i.e.:
\[
\mathbb{E}_{t_i}[x_i(t_i, \boldsymbol{t}_{-i}^{(1)})] = \int_{-\infty}^0 x_i(t_i, \boldsymbol{t}_{-i}^{(1)}) f_i(t_i) \, dt_i < 1.
\]
Define the attribution deficit as:
\[
\delta = g_i(\boldsymbol{t}_{-i}^{(1)}) \cdot (1 - \mathbb{E}_{t_i}[x_i(t_i, \boldsymbol{t}_{-i}^{(1)})]) > 0,
\]
where \( g_i(\boldsymbol{t}_{-i}^{(1)}) \) is the density of \( \boldsymbol{t}_{-i}^{(1)}\). This deficit \( \delta \) represents the portion of expected attribution not assigned at \( \boldsymbol{t}_{-i}^{(1)} \). To maintain \( \mathbb{E}_{\boldsymbol{t}}[x_i(t_i, \boldsymbol{t}_{-i})] = e_i \), this deficit must be allocated elsewhere. Since \( x_i^{*} \) assigns zero attribution for \( \max\{\boldsymbol{t}_{-i}\} > \theta_i \), we abstract a single \( \boldsymbol{t}_{-i}^{(2)} \) with \( \max\{\boldsymbol{t}_{-i}^{(2)}\} = m_2 > \theta_i \), such that the expected attribution at \( \boldsymbol{t}_{-i}^{(2)} \) is exactly:
\[
\mathbb{E}_{t_i}[x_i(t_i, \boldsymbol{t}_{-i}^{(2)})] \cdot g_i(\boldsymbol{t}_{-i}^{(2)}) = \delta = g_i(\boldsymbol{t}_{-i}^{(1)}) \cdot (1 - \mathbb{E}_{t_i}[x_i(t_i, \boldsymbol{t}_{-i}^{(1)})]).
\]
From the accuracy bound, the accuracy at \( m_1 \leq \theta_i \) is at most \( 1 - F_i(m_1) \geq 1 - F_i(\theta_i) \), while at \( m_2 > \theta_i \), it is at most \( 1 - F_i(m_2) \leq 1 - F_i(\theta_i) \). Construct a new mechanism \( x_i' \) by setting:
\[
x_i'(t_i, \boldsymbol{t}_{-i}^{(1)}) = 1, \quad x_i'(t_i, \boldsymbol{t}_{-i}^{(2)}) = 0,
\]
and keeping \( x_i'(t_i, \boldsymbol{t}_{-i}) = x_i(t_i, \boldsymbol{t}_{-i}) \) for all other \( \boldsymbol{t}_{-i} \). The expected attribution of \( x_i' \) is:
\[
\mathbb{E}_t[x_i'(t_i, \boldsymbol{t}_{-i})] = \mathbb{E}_t[x_i(t_i, \boldsymbol{t}_{-i})] - \delta + \delta = e_i,
\]
since the increase at \( \boldsymbol{t}_{-i}^{(1)} \) (\( g_i(\boldsymbol{t}_{-i}^{(1)}) \cdot (1 - \mathbb{E}_{t_i}[x_i(t_i, \boldsymbol{t}_{-i}^{(1)})]) \)) offsets the decrease at \( \boldsymbol{t}_{-i}^{(2)} \). Before the change, the attribution at \( \boldsymbol{t}_{-i}^{(1)} \), with expected amount \( g_i(\boldsymbol{t}_{-i}^{(1)}) \cdot \mathbb{E}_{t_i}[x_i(t_i, \boldsymbol{t}_{-i}^{(1)})] \), had an accuracy of at most \( 1 - F_i(m_1) \), and the \( \delta \)-portion at \( \boldsymbol{t}_{-i}^{(2)} \) had an accuracy of at most \( 1 - F_i(m_2) \). After the change, the total attribution at \( \boldsymbol{t}_{-i}^{(1)} \), now \( g_i(\boldsymbol{t}_{-i}^{(1)}) \cdot \mathbb{E}_{t_i}[x_i(t_i, \boldsymbol{t}_{-i}^{(1)})] + \delta = g_i(\boldsymbol{t}_{-i}^{(1)}) \), achieves an accuracy of exactly \( 1 - F_i(m_1) \). Since the original attribution at \( \boldsymbol{t}_{-i}^{(1)} \) had accuracy \( \leq 1 - F_i(m_1) \), the transferred \( \delta \)-portion must contribute an accuracy \( \geq 1 - F_i(m_1) \) to maintain the weighted average of \( 1 - F_i(m_1) \). Thus, the \( \delta \)-portion, previously at \( \boldsymbol{t}_{-i}^{(2)} \) with accuracy \( \leq 1 - F_i(m_2) \), now achieves accuracy \( \geq 1 - F_i(m_1) \). Since \( m_1 \leq \theta_i < m_2 \), we have:
\[
1 - F_i(m_1) \geq 1 - F_i(\theta_i) \geq 1 - F_i(m_2).
\]
Hence, the \( \delta \)-portion's accuracy increases from \( \leq 1 - F_i(m_2) \) to \( \geq 1 - F_i(m_1) \geq 1 - F_i(m_2) \), not decreasing the overall expected correct attribution. This \( x_i' \) is also optimal.

Thus, \( x_i^{*}(t_i, \boldsymbol{t}_{-i}) \) maximizes \( \mathbb{E}_t[x_i(t_i, \boldsymbol{t}_{-i}) \cdot \mathbb{I}[t_i > \max\{\boldsymbol{t}_{-i}\}]] \), satisfies DSIC, and meets \( \mathbb{E}_t[x_i^{*}(t_i, \boldsymbol{t}_{-i})] = e_i \) with \(  G_i(\theta_i) =e_i \). 
\end{proof}

\subsection{Proof of Theorem~\ref{thm:pvm-optimal-homog}}
\label{app:proof-pvm-optimal-homog}
\begin{proof}
We prove the optimality of PVM in the homogeneous setting in two steps. First, we show that PVM assigns the same threshold to all platforms. Then, we prove that this symmetric allocation rule is indeed optimal among all DSIC mechanisms.

In the homogeneous setting, each platform $i \in [n]$ under PVM is assigned an allocation rule:
\begin{equation}\label{app:eq_used_PVM_optimal}
x_i(\boldsymbol{t}) = 
\begin{cases}
1 & \text{if } \max\{\boldsymbol{t}_{-i}\} \leq \alpha_i \\ 
0 & \text{otherwise}
\end{cases}
\end{equation}
where the threshold $\alpha_i$ is chosen such that
\[
\prod_{j \neq i} F_j(\alpha_i) = \frac{1}{n}.
\]
Due to symmetry (i.e., all $F_j = F$), this reduces to
\[
n \cdot F(\alpha_i)^{n-1} = 1, \quad \forall i \in [n].
\]
Thus, all $\alpha_i$ are equal, and we denote the common threshold by $\alpha$, which satisfies
\[
F(\alpha) = \left(\frac{1}{n}\right)^{1/(n-1)}.
\]

We now argue that this allocation rule is optimal. Note that any optimal DSIC mechanism must allocate a total expected credit of 1; otherwise, unallocated credit does not contribute to attribution accuracy, and reallocating it (without violating DSIC) could only improve or preserve performance.

By Lemma~\ref{lem:threshold-best}, any DSIC mechanism can be reduced to a threshold-based one. Therefore, it suffices to consider threshold-based mechanisms that allocate full expected credit.

Suppose there exists an optimal threshold-based mechanism with $K \geq 2$ platforms using thresholds different from $\alpha$. Then, there must exist at least two platforms with thresholds on opposite sides of $\alpha$. Without loss of generality, let platform 1 use $\theta_1 < \alpha$ and platform 2 use $\theta_2 > \alpha$, with allocation rules:
\[
x_1^{(0)}(\boldsymbol{t}) = \mathbb{I}\left[\max_{j \ne 1} t_j \le \theta_1 \right], \quad 
x_2^{(0)}(\boldsymbol{t}) = \mathbb{I}\left[\max_{j \ne 2} t_j \le \theta_2 \right].
\]

Now consider reallocating attribution from platform 2 to platform 1 over the region where $\max_{j \ne 1} t_j \in (\alpha, \theta_2]$. Define the modified allocation:
\[
x_1^{(1)}(\boldsymbol{t}) = \mathbb{I}\left[ \max_{j \ne 1} t_j \le \theta_1 \text{ or } \max_{j \ne 1} t_j \in (\alpha, \theta_2] \right], \quad
x_2^{(1)}(\boldsymbol{t}) = \mathbb{I}\left[ \max_{j \ne 2} t_j \le \alpha \right].
\]

By symmetry, this transformation preserves both the expected total attribution and overall accuracy.

Applying Lemma~\ref{lem:threshold-best} again, there exists a new threshold $\theta_1'$ such that
\[
x_1^{(2)}(\boldsymbol{t}) = \mathbb{I}\left[ \max_{j \ne 1} t_j \le \theta_1' \right]
\]
achieves the same expected credit and accuracy as $x_1^{(1)}$.

If $\theta_1' = \alpha$, the number of platforms with thresholds different from $\alpha$ becomes $K-2$. Otherwise, it is reduced to $K-1$. In this latter case, there must exist another platform $i$ with threshold $\theta_i \ne \alpha$ such that $(\theta_1' - \alpha)(\theta_i - \alpha) < 0$, i.e., the two thresholds lie on opposite sides of $\alpha$. We can then repeat the above exchange process between platform 1 and platform $i$.

This iterative process terminates when all platforms use threshold $\alpha$ (i.e., $K = 0$), yielding a symmetric, DSIC mechanism with the same optimal accuracy. This is exactly the allocation rule implemented by PVM in the homogeneous setting. The proof is complete.
\end{proof}

\subsection{Proof of Theorem \ref{thm:pvm-n-homogeneous}}
\label{app:proof-pvm-n-homogeneous}
\begin{proof}
Since the platforms are homogeneous, the thresholds are identical, i.e., 
\[
\theta_1 = \theta_2 = \cdots = \theta_n = \theta.
\]
By definition, the threshold $\theta$ is chosen so that for any platform $i$,
\[
\prod_{j\neq i} F_j(\theta) = \int_{-\infty}^{0} f_i(t) \prod_{j\neq i} F_j(t) \, dt.
\]
Due to the symmetry, the probability that any individual platform's true display time satisfies 
\[
t_i \le \theta
\]
is given by
\[
P(t_i \le \theta) = \left( \frac{1}{n} \right)^{\frac{1}{n-1}}.
\]
Denote
\[
a = \left(\frac{1}{n}\right)^{\frac{1}{n-1}}.
\]
Then, by construction, the probability that the report times of the other $n-1$ platforms are all at most $\theta$ is
\[
a^{n-1} = \left[\left(\frac{1}{n}\right)^{\frac{1}{n-1}}\right]^{n-1} = \frac{1}{n}.
\]

Correct attribution occurs under two mutually exclusive events:
\begin{enumerate}
    \item \emph{All $n$ platforms report times no greater than $\theta$.} The probability of this event is 
    \[
    a^n.
    \]
    
    \item \emph{Exactly one platform fails to report a time at most $\theta$ (and hence $n-1$ platforms report times $\le \theta$).} There are $\binom{n}{1} = n$ ways to choose which platform exceeds $\theta$, and the probability for each such configuration is
    \[
    a^{n-1}(1 - a).
    \]
    Thus, the total probability of this event is 
    \[
    n \, a^{n-1}(1 - a).
    \]
\end{enumerate}

Hence, the overall accuracy of the Peer-Validated Mechanism is
\[
ACC(\mathcal{M}_{\text{LCM}};F^{\times n}) = a^n + n\, a^{n-1} (1 - a) = a^{n-1} \Bigl(a + n (1-a)\Bigr).
\]
Since
\[
a^{n-1} = \frac{1}{n},
\]
we have
\[
ACC(\mathcal{M}_{\text{LCM}};F^{\times n}) = \frac{1}{n} \Bigl(a + n (1-a)\Bigr) = 1 - \frac{n-1}{n}a.
\]
Substituting back $a = \left(\frac{1}{n}\right)^{\frac{1}{n-1}}$, we obtain the stated accuracy:
\[
ACC(\mathcal{M}_{\text{LCM}};F^{\times n}) = 1 - \Bigl(1 - \frac{1}{n}\Bigr)\left(\frac{1}{n}\right)^{\frac{1}{n-1}}.
\]

Therefore
\[
ACC(\mathcal{M}_{\text{LCM}}) = \inf_{F} ACC(\mathcal{M}_{\text{LCM}};F^{\times n}) = 1 - \Bigl(1 - \frac{1}{n}\Bigr)\left(\frac{1}{n}\right)^{\frac{1}{n-1}}
\]
\end{proof}

\subsection{Proof of Theorem \ref{thm:pvm-two-heterogeneous}}
\label{app:proof-pvm-two-heterogeneous}

\begin{proof}
We proceed in two parts: first, we formulate an optimization problem to derive the lower bound on accuracy; second, we construct specific distributions to demonstrate that this bound is achievable, thereby proving its tightness.

\subsubsection*{Part 1: Derivation of the Lower Bound via Optimization}

Consider two platforms, where \(p \in [0, 1]\) represents the probability that Platform 1 is the true last-displaying platform, and \(1 - p\) is the corresponding probability for Platform 2. Since the PVM is dominant-strategy incentive-compatible (DSIC), there exist thresholds \(\theta_1\) and \(\theta_2\) such that the allocation functions for Platforms 1 and 2 are defined as:
\[
x_1(t_2) = \begin{cases} 
1, & \text{if } t_2 \leq \theta_1, \\
0, & \text{otherwise},
\end{cases}
\qquad
x_2(t_1) = \begin{cases} 
1, & \text{if } t_1 \leq \theta_2, \\
0, & \text{otherwise}.
\end{cases}
\]
These thresholds satisfy \(P(t_2 \leq \theta_1) = p\) and \(P(t_1 \leq \theta_2) = 1 - p\). Without loss of generality, we assume \(\theta_1 \leq \theta_2\).

To model the display times, we partition the domain \((-\infty, 0]\) into three intervals for Platform 1: \((-\infty, \theta_1]\), \((\theta_1, \theta_2]\), and \((\theta_2, 0]\), with corresponding probabilities \(1 - p - x\), \(x\), and \(p\), where \(x \in [0, 1 - p]\). Similarly, for Platform 2, the same intervals have probabilities \(p\), \(y\), and \(1 - p - y\), respectively, where \(y \in [0, 1 - p]\).

When the display times \(t_1\) and \(t_2\) fall in different intervals, the relative ordering of \(t_1\) and \(t_2\) and the correctness of attribution are straightforward. However, when \(t_1\) and \(t_2\) lie within the same interval, we define \(b_1\), \(b_2\), and \(b_3 \in [0, 1]\) as the probabilities that \(t_1 > t_2\) in the intervals \((-\infty, \theta_1]\), \((\theta_1, \theta_2]\), and \((\theta_2, 0]\), respectively.

Since Platform 1 is the true last-displaying platform with probability \(p\), we apply the law of total probability to obtain:
\begin{align*}
p ={} & P(t_1 \in (\theta_2, 0]) P(t_2 \in (-\infty, \theta_2]) \\
&+ P(t_1 \in (\theta_1, \theta_2]) P(t_2 \in (-\infty, \theta_1]) \\
& + b_1 P(t_1 \in (-\infty, \theta_1]) P(t_2 \in (-\infty, \theta_1]) \\
&+ b_2 P(t_1 \in (\theta_1, \theta_2]) P(t_2 \in (\theta_1, \theta_2]) \\
& + b_3 P(t_1 \in (\theta_2, 0]) P(t_2 \in (\theta_2, 0]).
\end{align*}
This constraint ensures that the probability for Platform 1 is \(p\), while the complementary probability \(1 - p\) for Platform 2 is naturally satisfied, as the total probability across both platforms sums to 1.

The probability of misattribution, denoted \(False\), occurs when the PVM incorrectly attributes the last display, given by:
\begin{align*}
False ={} & P(t_1 \in (\theta_1, \theta_2]) P(t_2 \in (\theta_1, \theta_2]) P(t_1 > t_2 |t_1,t_2\in (\theta_1,\theta_2]) \\&
+ P(t_1 \in (\theta_2, 0]) P(t_2 \in (\theta_1, \theta_2]) \\
&+ P(t_1 \in (\theta_2, 0]) P(t_2 \in (\theta_2, 0]) \\
=& b_2 xy + p(1 - p).
\end{align*}
To determine the worst-case accuracy, we aim to maximize the misattribution probability through the following optimization problem:
\begin{align*}
\max \quad & b_2 xy + p(1 - p) \\
\text{s.t.} \quad & p = p(p + y) + xp + b_1 (1 - p - x)p + b_2 xy + b_3 p(1 - p - y), \\
& 0 \leq x \leq 1 - p, \\
& 0 \leq y \leq 1 - p, \\
& p \in [0, 1], \\
& b_1, b_2, b_3 \in [0, 1].
\end{align*}

To simplify this optimization, we first demonstrate that the contributions of \(b_1\) and \(b_3\) can be set to zero without loss of optimality, as they do not directly contribute to the misattribution probability.

\begin{lemma}
There exists an optimal solution where \(b_1 = b_3 = 0\).
\end{lemma}

\begin{proof}
Consider an optimal solution where \(b_1, b_3 > 0\) satisfies the constraint \(p = p(p + y) + xp + b_1 (1 - p - x)p + b_2 xy + b_3 p(1 - p - y)\), implying \(p(p + y) + xp + b_2 xy \leq p\). We analyze two cases:
\begin{itemize}
    \item \textbf{Case 1}: If there exists \(b_2'\) such that \(b_2\leq b_2' \leq 1\) and \(p(p + y) + xp + b_2' xy = p\), then setting \((b_1, b_2, b_3) = (0, b_2', 0)\) increases the objective function, as \(b_2' xy \geq b_2 xy\), while keeping feasible.
    \item \textbf{Case 2}: If \(p(p + y) + xp + xy < p\), note that \(p(p + y) + xp + xy = (p + x)(p + y)\). Since \((p + \hat{x})(p + \hat{y}) \in [p^2, 1]\) for \(\hat{x}, \hat{y} \in [0, 1 - p]\), there exist \(x\leq x' \leq 1-p\) and \(y\leq y' \leq 1-p\) such that \((p + x')(p + y') = p\). Setting \((x, y, b_1, b_2, b_3) = (x', y', 0, 1, 0)\) increases the objective function, as \(x'y' \geq b_2xy\). And it is still feasible.
\end{itemize}
Thus, an optimal solution exists with \(b_1 = b_3 = 0\).
\end{proof}

Applying this result, the constraint simplifies to:
\[
p = p^2 + p(x + y) + b_2 xy \quad \implies \quad b_2 xy = p - p^2 - p(x + y).
\]
Since \(b_2 \in [0, 1]\), we derive:
\[
p - p^2 - p(x + y) \geq 0 \implies x + y \leq 1 - p,
\]
\[
p - p^2 - p(x + y) \leq xy \implies p^2 + p(x + y) + xy \geq p.
\]

The optimization problem becomes:

\begin{align*}
\max \quad & 2(p - p^2) - p(x + y) \\
\text{s.t.} \quad & x + y \leq 1 - p, \\
& p^2 + p(x + y) + xy \geq p, \\
& x \leq 1 - p, \\
& y \leq 1 - p, \\
& x, y \geq 0.
\end{align*}

To further simplify, we explore whether symmetry in the variables \(x\) and \(y\) can be assumed, as this may streamline the optimization by reducing the number of variables.

\begin{lemma}
There exists an optimal solution where \(x = y\).
\end{lemma}

\begin{proof}
Suppose an optimal solution has \(x^* + y^* = S\) with \(x^* \neq y^*\). Define \(x' = y' = \frac{S}{2}\). The objective function \(2(p - p^2) - p(x + y)\) remains unchanged, as \(x' + y' = S\). By the arithmetic mean-geometric mean (AM-GM) inequality, \(x'y' \geq x^* y^*\), and all constraints are satisfied. Hence, an optimal solution exists with \(x = y\).
\end{proof}

Assuming \(x = y\), the optimization problem simplifies to:

\begin{align*}
\max \quad & 2(p - p^2) - 2px \\
\text{s.t.} \quad & x \leq \frac{1 - p}{2}, \\
& p^2 + 2px + x^2 \geq p, \\
& x \leq 1 - p, \\
& x \geq 0, \\
& p \in [0, 1].
\end{align*}

The constraint \(p^2 + 2px + x^2 \geq p\) implies \((p + x)^2 \geq p\), yielding \(x \geq \sqrt{p} - p\) (discarding the infeasible root \(p + x \leq -\sqrt{p}\)). Thus, the problem reduces to:

\begin{align*}
\max \quad & 2p(1 - p - x) \\
\text{s.t.} \quad & \sqrt{p} - p \leq x \leq \frac{1 - p}{2}, \\
& p \in [0, 1].
\end{align*}

Since \(\frac{1 - p}{2} \geq \sqrt{p} - p\) for all \(p \in [0, 1]\) (as \(1 + p \geq 2\sqrt{p}\)), the objective function is maximized when \(x = \sqrt{p} - p\). Therefore, the problem becomes:

\begin{align*}
\max \quad & 2p(1 - \sqrt{p}) \\
\text{s.t.} \quad & p \in [0, 1].
\end{align*}

Define the function \(g(p) = 2p(1 - \sqrt{p})\). Its derivative is \(g'(p) = 2 - 3\sqrt{p}\), and setting \(g'(p) = 0\) yields \(p = \frac{4}{9}\). Since \(g(p)\) is increasing on \((0, \frac{4}{9})\) and decreasing on \((\frac{4}{9}, 1)\), the maximum occurs at \(p = \frac{4}{9}\), where:
\[
g\left(\frac{4}{9}\right) = 2 \cdot \frac{4}{9} \cdot \left(1 - \sqrt{\frac{4}{9}}\right) = \frac{8}{27}.
\]

Thus, the worst-case probability of misattribution is \(\frac{8}{27}\), and the corresponding worst-case accuracy is:
\[
1 - \frac{8}{27} = \frac{19}{27}.
\]

\subsubsection*{Part 2: Establishing the Tightness of the Bound}

To demonstrate that the accuracy bound of \(\frac{19}{27}\) is tight, we construct explicit distributions for Platforms 1 and 2. Define the probability density function for Platform 1 as:
\[
f_1(t) = \begin{cases} 
1, & t \in \left(-\frac{22}{9}, -2\right] \cup \left(-\frac{13}{9}, -\frac{11}{9}\right] \cup \left(-\frac{1}{3}, 0\right], \\
0, & \text{otherwise},
\end{cases}
\]
and for Platform 2 as:
\[
f_2(t) = \begin{cases} 
1, & t \in \left(-\frac{25}{9}, -\frac{22}{9}\right] \cup \left(-\frac{11}{9}, -1\right] \cup \left(-\frac{7}{9}, -\frac{1}{3}\right], \\
0, & \text{otherwise}.
\end{cases}
\]

By evaluating the probability of correct attribution under the PVM with these distributions, we confirm that the accuracy is precisely \(\frac{19}{27}\). This result establishes that the lower bound derived in Part 1 is achievable.

In conclusion, we have rigorously demonstrated that the worst-case accuracy of the PVM for two platforms with arbitrary probability density functions \(f_1\) and \(f_2\) over \((-\infty, 0]\) is \(\frac{19}{27}\). By constructing specific distributions that achieve this accuracy, we have proven that this bound is tight.
\end{proof}

\subsection{Proof of Theorem \ref{thm:pvm-n-heterogeneous}}
\label{app:proof-pvm-n-heterogeneous}
\begin{proof}
We aim to lower bound the worst-case accuracy of PVM for $n$ platforms with arbitrary independent display-time distributions \( f_i(t) \) supported on $(-\infty, 0]$. To this end, we construct a tree-based mechanism $\mathcal{M}_{\text{tree}}$, show that it is DSIC and matches PVM in expected attribution, and then prove that PVM achieves at least the same accuracy. Finally, we lower bound the accuracy of $\mathcal{M}_{\text{tree}}$ via induction.

\paragraph{Step 1: Definition of the Tree-Based Mechanism.} 
Let \( d = \lceil \log_2 n \rceil \), and construct a complete binary tree of depth \( d \): the root node \( S_0 \)contains all platforms,\( S_0 = \{1, \dots, n\} \); each internal node splits its platform subset into disjoint halves \( L \) and \( R \); leaf nodes are singletons \( \{i\} \), padded if necessary.

Given report profile \( \boldsymbol{r} \), define the eligible set \( S = \{i \mid r_i \leq 0\} \). At each internal node with children \( L \) and \( R \), define thresholds:
\[
P\left( \max_{j \in R \cap S} t_j \leq \theta_L \right) = P\left( \text{last click} \in L \mid \text{last click} \in (L \cup R) \cap S \right),
\]
and similarly for \( \theta_R \). If the corresponding report maxima satisfy the threshold condition, recursion continues down that child; otherwise, attribution halts. At a leaf \( \{i\} \), assign
\[
x_i^{\text{tree}}(\boldsymbol{r}) = \mathbb{I}[r_i \leq 0].
\]

\paragraph{Incentive Compatibility.}
The mechanism is DSIC: reporting \( r_i > t_i \) risks \( r_i > 0 \), which sets \( x_i = 0 \), while under-reporting is infeasible. Thus, truthful reporting \( r_i = t_i \) is dominant, and we may assume \( \boldsymbol{r} = \boldsymbol{t} \) in analysis.

\paragraph{Expected Attribution.}
We show by induction that \( \mathbb{E}[x_i^{\text{tree}}] = P(i = \arg\max_j t_j) \).  \textbf{Base case:} At the root, all probability mass is preserved. \textbf{Inductive step:} At each node, the threshold test partitions the probability mass proportionally to the conditional probability of being the last click in each child. Hence, the correct mass propagates recursively to the leaf corresponding to the true last-click platform.

\paragraph{Step 2: PVM Dominates the Tree Mechanism.}
PVM is also DSIC and satisfies:
\[
\mathbb{E}[x_i^{\text{PVM}}] = P(i = \arg\max_j t_j).
\]
By Lemma~\ref{lem:threshold-best}, among all DSIC mechanisms with fixed expected attribution \( e_i \), the single-threshold rule
\[
x_i(t_i, \boldsymbol{t}_{-i}) = \mathbb{I}\left[ \max_{j \neq i} t_j \leq \theta_i \right],
\]
where \( G_i(\theta_i) = e_i \), is an optimal attribution rule. Since PVM uses this optimal structure, and $\mathcal{M}_{\text{tree}}$ achieves the same $e_i$, it follows:
\[
\mathbb{E}[x_i^{\text{PVM}} \cdot \mathbb{I}[i = \arg\max_j t_j]] \geq \mathbb{E}[x_i^{\text{tree}} \cdot \mathbb{I}[i = \arg\max_j t_j]].
\]
Summing over \( i \) yields:
\[
\text{ACC}(\text{PVM}) \geq \text{ACC}(\mathcal{M}_{\text{tree}}).
\]

\paragraph{Step 3: Bounding Tree Mechanism Accuracy.}
Let \( E_k \) denote the event that the true last-click platform survives after \( k \) tree levels. Clearly, \( P(E_0) = 1 \). 
At each level of the tree, the mechanism compares two disjoint subsets of platforms---denoted $L$ and $R$---based on the maximum reported times within each group. This effectively reduces the comparison to a two-platform setting, where the “virtual platforms” are represented by $\max_{j \in L} t_j$ and $\max_{j \in R} t_j$. Thus, a PVM-style validation test can be applied between these two groups, using thresholds that align with the probability that the true last click lies in one group versus the other. By Theorem~\ref{thm:pvm-two-heterogeneous}, each such binary test preserves the true platform with probability at least \( \frac{19}{27} \). Thus:
\[
 P(E_k) = P(E_{k-1})\cdot P(E_{k-1}|E_k) \geq P(E_{k-1})\cdot \frac{19}{27}
\]
\[
P(E_d) \geq \left( \frac{19}{27} \right)^d = \left( \frac{19}{27} \right)^{\lceil \log_2 n \rceil}.
\]
This gives:
\[
\text{ACC}(\mathcal{M}_{\text{tree}}) \geq \left( \frac{19}{27} \right)^{\lceil \log_2 n \rceil}.
\]

\paragraph{Conclusion.}
Combining the steps:
\[
\text{ACC}(\text{PVM}) \geq \text{ACC}(\mathcal{M}_{\text{tree}}) \geq \left( \frac{19}{27} \right)^{\lceil \log_2 n \rceil}.
\]
\end{proof}

\subsection{Proof of Proposition~\ref{prop:pvm-is-fair}}
Since PVM is DSIC, we get 
\[
 \mathbb{E}_{\boldsymbol{t}}[x_i(\boldsymbol{t})] = \mathbb{E}_{\boldsymbol{t}}[\max_{j\neq i}\{t_j\}\leq \alpha_i] = \prod_{j\neq i}F_j(\alpha_i)= P(i = \arg\max_j \{t_j\}),
 \]
which means that PVM is fair.

\subsection{Proof of Proposition~\ref{prop:lc-is-not-fair}} %
\label{app:proof-last-click-fairness-detail}
\begin{proof}
We analyze the worst-case fairness of the Last-Click Mechanism (\( \mathcal{M}_{LCM} \)) under different settings. Recall that for a mechanism \( \mathcal{M} \) and distribution profile \( \boldsymbol{F} \), the fairness is defined as
\[
\text{Fair}(\mathcal{M}; \boldsymbol{F}) = \min_{\{i \mid P(i = \arg\max_j t_j) > 0\}} \left\{\frac{\mathbb{E}[x_i]}{P(i = \arg\max_j t_j)}\right\}.
\]

\paragraph{Homogeneous case: \( n = 2 \).}
Let \( F \) be any distribution on \( (-\infty, 0] \) such that the worst-case equilibrium under LCM is symmetric with delay \( \tau_0 \). From Appendix~\ref{app:proof-last-click-thm1}, the probability of an invalid report is \( p = 1 - F(-\tau_0) = \sqrt{2} - 1 \). The probability that both reports are invalid is \( p^2 \), and attribution occurs with probability \( 1 - p^2 \).

By symmetry, each platform receives attribution with probability \( \frac{1}{2}(1 - p^2) \), and has last-click probability \( \frac{1}{2} \). Thus, the fairness is
\[
\text{Fair}(\mathcal{M}_{LCM}; F^{\times 2}) = \frac{\frac{1}{2}(1 - p^2)}{1/2} = 1 - (\sqrt{2} - 1)^2 \approx 0.828.
\]

\paragraph{Homogeneous case: general \( n \geq 3 \).}
Let \( (\tau_0, \dots, \tau_0) \) be the worst-case symmetric equilibrium for \( n \) identical platforms. Let \( p = 1 - F(-\tau_0) \) be the probability of an invalid report. From Theorem~\ref{thm:last-click-n-homogeneous}, we have:
\[
\left(\sqrt[3]{\frac{2+\sqrt{6}}{4}} + \sqrt[3]{\frac{2-\sqrt{6}}{4}}\right)^2 \le p < \left(\frac{1}{n}\right)^{1/(n-1)}.
\]

Hence, attribution occurs with probability \( 1 - p^n \), and each platform receives expected credit \( \frac{1}{n}(1 - p^n) \), while the last-click probability remains \( \frac{1}{n} \). Therefore,
\[
\text{Fair}(\mathcal{M}_{LCM}; F^{\times n}) = \frac{\frac{1}{n}(1 - p^n)}{1/n} = 1 - p^n.
\]
Using the bounds on \( p \), we get:
\[
1 - \left(\frac{1}{n}\right)^{n/(n-1)} < \text{Fair}(\mathcal{M}_{LCM}; F^{\times n}) \le 1 - \left(\left(\sqrt[3]{\frac{2+\sqrt{6}}{4}} + \sqrt[3]{\frac{2-\sqrt{6}}{4}}\right)^2\right)^n.
\]

\paragraph{Heterogeneous case: \( n \geq 2 \).}
In Appendix~\ref{app:proof-pvm-n-heterogeneous}, we construct a heterogeneous instance where one platform (platform 1) has a concentrated distribution supported strictly before the supports of all others. In equilibrium, platform 1 always reports near 0 and consistently secures the attribution. Meanwhile, all other platforms \( k \in \{2, \dots, n\} \) have equal positive probability \( P(k = \arg\max_j t_j) = \frac{1}{n-1} \), but receive zero expected credit under LCM.

Thus, for each such platform \( k \), fairness approaches 0, and the overall worst-case fairness is:
\[
\inf_{\boldsymbol{F}} \text{Fair}(\mathcal{M}_{LCM}; \boldsymbol{F}) = 0.
\]
\end{proof}

\section{Numerical Experiments Details}
\label{app:simulation_detail}

This section empirically validates our theoretical findings by comparing the Peer-Validated Mechanism (PVM) with the Last-Click Mechanism (LCM) using simulations grounded in real-world click-timing data. 

\paragraph{Empirical Click-Timing Distributions}
We obtained realistic click time distributions using click records from four industrial advertising platforms (anonymized in this study as A, B, C, and D). The underlying dataset, processed to ensure user privacy,  featuring timestamps in seconds within the $[-100, 0]$ interval (relative to conversion).\footnote{The $[-100, 0]$-second interval for input data was selected because click data for $t < -100s$ was observed to be very sparse and exhibited a highly discrete distribution; its inclusion could impair the initial accuracy of fitting continuous density functions.} For each platform, an empirical click-time probability density was derived as follows:
\begin{enumerate}
    \item \textbf{Kernel Density Estimation (KDE):} Each platform's click-time density was initially estimated from its raw data (filtered to the $[-100, 0]$-second interval) via Gaussian KDE. The bandwidth for the KDE was set according to Scott‘s rule, a common and robust method for automatic bandwidth determination in KDE. This non-parametric method captures unique distributional characteristics without imposing restrictive assumptions on the density's functional form.
    \item \textbf{PDF Support Finalization and Normalization:} To incorporate smoothed tails generated by the KDE process (based on the $[-100, 0]s$ input data) while ensuring click times $t_i \leq 0$, the operational support for these densities was defined as the interval $[-120, 0]$ seconds. Any probability mass estimated by KDE outside this $[-120, 0]$ range was discarded. The resulting functions, now defined on $[-120, 0]$, were subsequently re-normalized over this specific interval to ensure they integrate to one, forming valid probability densities.
\end{enumerate}
This procedure produces a probability density function (PDF) for each of the platforms A, B, C, and D. We denote the corresponding cumulative distribution functions (CDFs) by \(F_A, F_B, F_C,\) and \(F_D\), respectively. The resulting PDFs are illustrated in Figure~\ref{fig:Normalized PDFs for 4 Platforms}.

\begin{figure}[h!]
    \centering
    \includegraphics[width=1\linewidth]{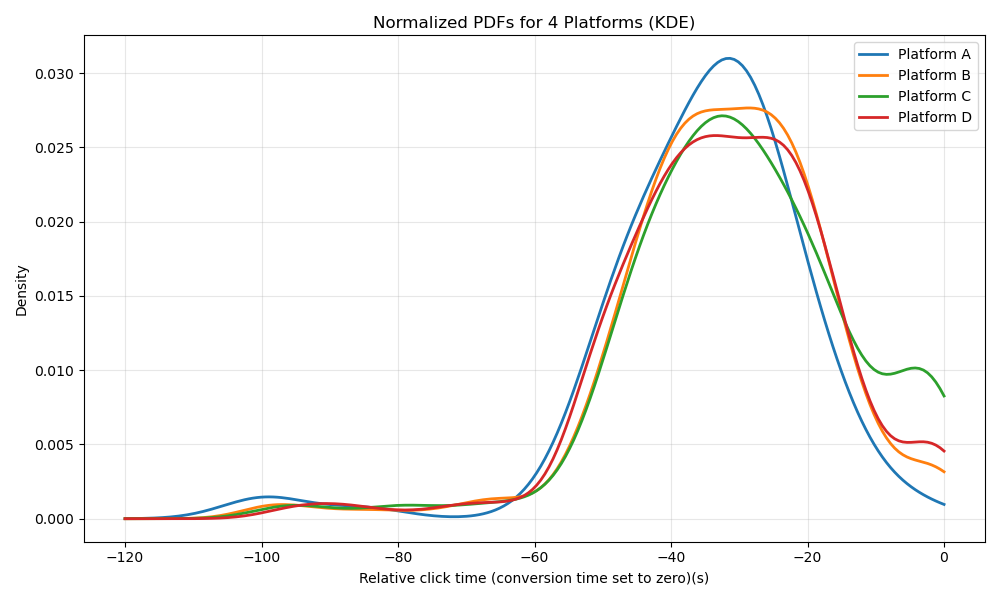}
    \caption{Normalized PDFs for 4 platforms}
    \label{fig:Normalized PDFs for 4 Platforms}
\end{figure}

\paragraph{Experimental Setup}
We compare the LCM and PVM mechanisms. Under LCM, platforms strategically report delays $\tau_i^*$, representing pure strategy Nash equilibria derived from the empirical click time distributions ($F_A, F_B, F_C, F_D$) established previously. Under PVM, platforms report truthfully ($\tau_i = 0$), consistent with its dominant-strategy incentive compatible (DSIC) properties. Each configuration was evaluated using $5 \times 10^4$ simulated user paths, repeated over 10 independent runs.

Two primary settings were considered, based on these empirical distributions:
\begin{enumerate}
    \item \textbf{Homogeneous Setting:} For each of the four distinct empirical distributions ($F_A, F_B, F_C, F_D$), we simulated scenarios with $n \in \{2,3,4,5\}$ identical platforms, all drawing click times from the chosen baseline distribution. This setting evaluates mechanism performance against an increasing number of statistically similar competitors.
    \item \textbf{Heterogeneous Setting:} We simulated $n=2$ platforms using all six pairwise combinations of the distinct empirical distributions ($F_A, F_B, F_C, F_D$). This tests mechanism robustness to diverse click timing behaviors.
\end{enumerate}
Performance was assessed by attribution accuracy ($ACC(\mathcal{M}; \boldsymbol{F})$) and fairness ($FAIR(\mathcal{M}; \boldsymbol{F})$).

The experiments were conducted using Python 3.8.19, running on a CPU. For mathematical computations, we utilized the SciPy 1.10.1 and NumPy 1.23.5 libraries. The total computation time for the experiments was approximately twenty minutes.

\paragraph{Results and Analysis}
The detailed performance metrics for individual configurations in the homogeneous and heterogeneous settings are presented in Table~\ref{tab:platform_n_empirical_anon} and Table~\ref{tab:pairwise_updated_empirical}, respectively.

\begin{table}[h!]
  \centering
 \caption{Summary of homogeneous results across four real-world platform distributions. Each entry reports the average LCM and PVM accuracy, PVM's absolute accuracy gain (mean~$\pm$~standard deviation across platform types), and improvement in fairness .}
  \label{tab:platform_n_empirical_anon}
  \small
  \resizebox{\textwidth}{!}{%
  \begin{tabular}{llrrrrrr}
    \toprule
    Plat. & $n$ & $\tau^*$ & $\alpha$ & LCM Acc.\ ($\mu\pm\sigma$) & LCM Fair.\ ($\mu\pm\sigma$) & PVM Acc.\ ($\mu\pm\sigma$) & PVM Fair.\ ($\mu\pm\sigma$) \\
    \midrule
    A & 2 & 20.78 & -33.91 & $0.7377 \pm 0.0018$ & $0.9769 \pm 0.0021$ & $0.7493 \pm 0.0019$ & $0.9967 \pm 0.0041$ \\
     & 3 & 25.16 & -31.40 & $0.4405 \pm 0.0014$ & $0.9767 \pm 0.0031$ & $0.6149 \pm 0.0019$ & $0.9947 \pm 0.0060$ \\
     & 4 & 25.44 & -29.70 & $0.3215 \pm 0.0030$ & $0.9834 \pm 0.0031$ & $0.5285 \pm 0.0017$ & $0.9939 \pm 0.0072$ \\
     & 5 & 27.45 & -28.42 & $0.1640 \pm 0.0017$ & $0.9817 \pm 0.0050$ & $0.4651 \pm 0.0022$ & $0.9897 \pm 0.0081$ \\
     \hline
    B & 2 & 18.15 & -31.83 & $0.7281 \pm 0.0021$ & $0.9753 \pm 0.0028$ & $0.7496 \pm 0.0018$ & $0.9985 \pm 0.0029$ \\
     & 3 & 21.28 & -29.03 & $0.4827 \pm 0.0019$ & $0.9829 \pm 0.0034$ & $0.6150 \pm 0.0019$ & $0.9964 \pm 0.0050$ \\
     & 4 & 23.43 & -27.13 & $0.2848 \pm 0.0023$ & $0.9811 \pm 0.0057$ & $0.5268 \pm 0.0023$ & $0.9894 \pm 0.0044$ \\
     & 5 & 24.53 & -25.72 & $0.1684 \pm 0.0012$ & $0.9807 \pm 0.0079$ & $0.4652 \pm 0.0017$ & $0.9964 \pm 0.0079$ \\
     \hline
    C & 2 & 17.52 & -31.03 & $0.6513 \pm 0.0020$ & $0.9573 \pm 0.0028$ & $0.7502 \pm 0.0016$ & $0.9953 \pm 0.0034$ \\
     & 3 & 21.15 & -28.11 & $0.4039 \pm 0.0018$ & $0.9761 \pm 0.0021$ & $0.6160 \pm 0.0014$ & $0.9968 \pm 0.0047$ \\
     & 4 & 24.34 & -26.02 & $0.2013 \pm 0.0012$ & $0.9761 \pm 0.0031$ & $0.5284 \pm 0.0027$ & $0.9962 \pm 0.0060$ \\
     & 5 & 26.19 & -24.39 & $0.0956 \pm 0.0011$ & $0.9724 \pm 0.0060$ & $0.4643 \pm 0.0024$ & $0.9884 \pm 0.0090$ \\
     \hline
    D & 2 & 17.81 & -32.08 & $0.7195 \pm 0.0016$ & $0.9757 \pm 0.0014$ & $0.7492 \pm 0.0015$ & $0.9959 \pm 0.0031$ \\
     & 3 & 20.40 & -29.06 & $0.5016 \pm 0.0023$ & $0.9833 \pm 0.0045$ & $0.6158 \pm 0.0027$ & $0.9973 \pm 0.0040$ \\
     & 4 & 22.06 & -27.01 & $0.3262 \pm 0.0015$ & $0.9856 \pm 0.0059$ & $0.5277 \pm 0.0032$ & $0.9941 \pm 0.0078$ \\
     & 5 & 22.85 & -25.50 & $0.2159 \pm 0.0019$ & $0.9833 \pm 0.0059$ & $0.4655 \pm 0.0015$ & $0.9941 \pm 0.0080$ \\
    \bottomrule
  \end{tabular}%
  }
\end{table}
\begin{table}[h!]
  \centering
\caption{Heterogeneous setting: average results over all six pairwise combinations of fitted platform distributions. Shown are LCM and PVM accuracies, PVM's absolute accuracy improvement, and fairness improvement, reported as mean~$\pm$~standard deviation across pairs.}
  \label{tab:pairwise_updated_empirical}
  \small
  \resizebox{\textwidth}{!}{%
  \begin{tabular}{llrrrrrrrr}
    \toprule
    P\,1 & P\,2 & $\tau_1^*$ & $\tau_2^*$ & $\alpha_1$ & $\alpha_2$ & LCM Acc.\ ($\mu\pm\sigma$) & LCM Fair.\ ($\mu\pm\sigma$) & PVM Acc.\ ($\mu\pm\sigma$) & PVM Fair.\ ($\mu\pm\sigma$) \\
    \midrule
    A & B & 21.03 & 17.87 & -33.65 & -32.29 & $0.6887 \pm 0.0013$ & $0.8692 \pm 0.0032$ & $0.7519 \pm 0.0012$ & $0.9983 \pm 0.0024$ \\
    A & C & 21.03 & 18.47 & -33.71 & -31.57 & $0.6506 \pm 0.0014$ & $0.7691 \pm 0.0025$ & $0.7528 \pm 0.0011$ & $0.9961 \pm 0.0027$ \\
    A & D & 20.73 & 17.87 & -33.80 & -32.47 & $0.6911 \pm 0.0013$ & $0.8494 \pm 0.0027$ & $0.7510 \pm 0.0016$ & $0.9982 \pm 0.0023$ \\
    B & C & 17.27 & 18.77 & -31.96 & -30.92 & $0.6752 \pm 0.0022$ & $0.8600 \pm 0.0033$ & $0.7499 \pm 0.0011$ & $0.9968 \pm 0.0032$ \\
    B & D & 17.87 & 18.17 & -31.91 & -31.98 & $0.7336 \pm 0.0013$ & $0.9548 \pm 0.0015$ & $0.7500 \pm 0.0015$ & $0.9978 \pm 0.0027$ \\
    C & D & 18.47 & 16.52 & -30.95 & -32.10 & $0.6741 \pm 0.0015$ & $0.8885 \pm 0.0029$ & $0.7506 \pm 0.0019$ & $0.9995 \pm 0.0040$ \\
    \bottomrule
  \end{tabular}%
  }
\end{table}

To further assess PVM's advantage, Table~\ref{tab:homogeneous_acc_fair_improvement} reports results in the homogeneous setting across different values of $n$. For each case, we show the average accuracy of LCM and PVM, PVM's absolute accuracy improvement over LCM (with standard deviation across platform distributions), and the corresponding improvement in fairness.

\begin{table}[h!]
  \centering
  \caption{Average platform performance in the homogeneous setting. Columns report average accuracy for LCM and PVM, PVM's absolute accuracy improvement over LCM (mean and standard deviation across platforms), and improvement in fairness.}
  \small
  \begin{tabular}{crrrrrr}
    \toprule
    $n$ & Avg LCM Acc. & Avg PVM Acc. & Abs.\ Increase & Std(Abs.\ Inc) & Fair.\ Improve & Std(Fair.\ Imp) \\
    \midrule
    2 & 0.7091 & 0.7496 & 0.0404 & 0.0396 & 0.0248 & 0.0089 \\
    3 & 0.4572 & 0.6154 & 0.1583 & 0.0439 & 0.0157 & 0.0034 \\
    4 & 0.2834 & 0.5278 & 0.2444 & 0.0580 & 0.0107 & 0.0047 \\
    5 & 0.1610 & 0.4650 & 0.3041 & 0.0490 & 0.0111 & 0.0040 \\
    \bottomrule
  \end{tabular}
  \label{tab:homogeneous_acc_fair_improvement}
\end{table}

Additionally, Table~\ref{tab:hetero-sum-avg-pairs} provides an aggregate analysis of PVM's improvements over LCM in both accuracy and fairness for the heterogeneous setting.

\begin{table}[h!]
  \centering
  \caption{Heterogeneous setting: average results over all six pairwise combinations of fitted platform distributions. Shown are LCM and PVM accuracies, PVM's absolute accuracy improvement, and fairness improvement, reported as mean~$\pm$~standard deviation across pairs.}

  \label{tab:hetero-sum-avg-pairs}
  \small
  \begin{tabular}{cc}
    \toprule
    Metric & Value ($\mu\pm\sigma$) \\
    \midrule
    Avg LCM Acc.    & $0.6855 \pm 0.0276$ \\
    Avg PVM Acc.   & $0.7510 \pm 0.0011$ \\
    Abs. Increase  & $0.0655 \pm 0.0283$ \\
    Fair. Improve  & $0.1320 \pm 0.0598$ \\
    \bottomrule
  \end{tabular}
\end{table}

These simulation findings (Tables~\ref{tab:platform_n_empirical_anon}-\ref{tab:hetero-sum-avg-pairs}) consistently highlight the advantages of the PVM mechanism.
Regarding attribution accuracy, PVM significantly outperformed LCM across all scenarios. In homogeneous settings (Table~\ref{tab:homogeneous_acc_fair_improvement}), PVM's average absolute accuracy improvement over LCM grew substantially with the number of platforms $n$, increasing from 0.0404 for $n=2$ to 0.3041 for $n=5$. In heterogeneous pairings (Table~\ref{tab:hetero-sum-avg-pairs}), PVM achieved an average absolute accuracy improvement of 0.0655  over LCM, with the maximum absolute improvement reaching 0.1022 (Platform A and C). This demonstrates PVM's robustness against both an increasing number of competitors and diversity in platform click-timing distributions.

In terms of fairness, PVM consistently achieved scores very near 1 across all settings (Tables~\ref{tab:platform_n_empirical_anon}, \ref{tab:pairwise_updated_empirical}). In homogeneous settings, while LCM's fairness values were also high, PVM consistently reduced this deviation further (e.g., to approximately 0.0248 for $n=2$), thereby achieving fairness values even closer to the ideal. The superiority of PVM in fairness was particularly evident in heterogeneous settings. For example, in the A-C pairing (Table~\ref{tab:pairwise_updated_empirical}), PVM recorded a fairness of 0.9961, extremely close to 1, whereas LCM's fairness was 0.7691. As summarized in Table~\ref{tab:hetero-sum-avg-pairs}, PVM yielded an average absolute improvement in this fairness measure of 0.1321, with the maximum absolute improvement being 0.2270. This clearly shows PVM's greater resilience in maintaining equitable attribution under diverse conditions.